%% file: paper.tex
\documentclass[envcountsame,a4paper,runningheads]{llncs}

\usepackage[T1]{fontenc}
\usepackage{amsmath,amssymb}
\usepackage{stmaryrd}
\usepackage{mathtools}
\usepackage{todonotes}
\usepackage{xspace}
\usepackage[basic,bold]{complexity}
\usepackage{thmtools,thm-restate}
\usepackage{tikz}
\usetikzlibrary{automata, positioning,arrows,shapes}
\usepackage[hidelinks]{hyperref}
\hypersetup{
  colorlinks   = true,
  urlcolor     = blue,
  linkcolor    = blue,
  citecolor    = red
}
\usepackage{lineno}

\newcommand*\patchAmsMathEnvironmentForLineno[1]{%
  \expandafter\let\csname old#1\expandafter\endcsname\csname #1\endcsname
  \expandafter\let\csname oldend#1\expandafter\endcsname\csname end#1\endcsname
  \renewenvironment{#1}%
     {\linenomath\csname old#1\endcsname}%
     {\csname oldend#1\endcsname\endlinenomath}}% 
\newcommand*\patchBothAmsMathEnvironmentsForLineno[1]{%
  \patchAmsMathEnvironmentForLineno{#1}%
  \patchAmsMathEnvironmentForLineno{#1*}}%
\AtBeginDocument{%
\patchBothAmsMathEnvironmentsForLineno{equation}%
\patchBothAmsMathEnvironmentsForLineno{align}%
\patchBothAmsMathEnvironmentsForLineno{flalign}%
\patchBothAmsMathEnvironmentsForLineno{alignat}%
\patchBothAmsMathEnvironmentsForLineno{gather}%
\patchBothAmsMathEnvironmentsForLineno{multline}%
}
\let\doendproof\endproof
\renewcommand\endproof{~\hfill$\qed$\doendproof}
\title{The Geometry of Reachability in Continuous Vector Addition Systems with States}
\titlerunning{The Geometry of Reachability in Continuous VASS}
\author{Shaull Almagor\inst{1}
\and Arka Ghosh\inst{2}\thanks{A. Ghosh was partially supported by the NCN grant 2019/35/B/ST6/02322.} \and Tim Leys\inst{3}%\orcidID{0000-0001-6281-9693}
\and
Guillermo A. P\'erez\inst{3}%\orcidID{}
}
\authorrunning{S. Almagor et al.}
\institute{
Technion - Israel Institute of Technology, Israel\\
\and University of Warsaw, Poland\\
\and University of Antwerp -- Flanders Make, Belgium\\
\email{\{tim.leys, guillermo.perez\}@uantwerpen.be}}

\input{macros}

\begin{document}

\maketitle

\begin{abstract}
  We study the geometry of reachability sets of continuous vector addition
  systems with states (VASS). In particular
  we establish that they are ``almost'' Minkowski sums of convex cones and
  zonotopes generated by the vectors labelling the transitions of the VASS. We
  use the latter to prove that short so-called linear path schemes suffice as witnesses of reachability in continuous VASS of fixed dimension. Then, we give new polynomial-time algorithms for the reachability problem for linear path schemes. Finally, we also establish that
  enriching the model with zero tests makes the reachability problem intractable
  already for linear path schemes of dimension two.
\end{abstract}

\input{paper_sections/intro}

\input{paper_sections/preliminaries}

\input{paper_sections/gengeometry}

\input{paper_sections/QReach}

\input{paper_sections/QposReach}

\input{paper_sections/QReachZero}

\input{paper_sections/conclusion}

\bibliographystyle{splncs04}
\bibliography{bibliography}

\clearpage
\appendix
\input{paper_sections/geometry_missing_proofs}

\input{paper_sections/QReach_missing}

\input{paper_sections/QposReach_missing_proofs.tex}

\input{paper_sections/encodingminskymachines}

\end{document}

%% file: macros.tex
\newcommand{\calV}{\mathcal{V}}

\newcommand{\normtwo}[1]{\left\lVert #1 \right\rVert_2}

\newcommand{\Reach}{\mathrm{Reach}}

\newclass{\ETR}{ETR}

\newcommand{\interior}{\mathrm{int}}

\newcommand{\closure}{\mathrm{cl}}

\newcommand{\boundary}{\mathrm{bd}}

\renewcommand{\vec}[1]{\boldsymbol{#1}}

\newcommand{\cone}{\mathrm{cone}}
\newcommand{\linspan}{\mathrm{span}}

\newcommand{\convexhull}{\mathrm{conv}}

\newcommand{\zono}{\mathrm{zono}}

\newcommand{\adj}{\mathrm{adj}}

\newcommand{\aff}{\mathrm{aff}}

\newcommand{\relint}{\mathrm{relint}}

\makeatletter
\newcommand*{\da@rightarrow}{\mathchar"0\hexnumber@\symAMSa 4B }
\newcommand*{\xdashrightarrow}[2][]{%
  \mathrel{%
    \mathpalette{\da@xarrow{#1}{#2}{}\da@rightarrow{\,}}{}%
  }%
}
\newcommand*{\da@xarrow}[7]{%
  \sbox0{$\ifx#7\scriptstyle\scriptscriptstyle\else\scriptstyle\fi#5#1#6\m@th$}%
  \sbox2{$\ifx#7\scriptstyle\scriptscriptstyle\else\scriptstyle\fi#5#2#6\m@th$}%
  \sbox4{$#7\dabar@\m@th$}%
  \dimen@=\wd0 %
  \ifdim\wd2 >\dimen@
    \dimen@=\wd2 %   
  \fi
  \count@=2 %
  \def\da@bars{\dabar@\dabar@}%
  \@whiledim\count@\wd4<\dimen@\do{%
    \advance\count@\@ne
    \expandafter\def\expandafter\da@bars\expandafter{%
      \da@bars
      \dabar@ 
    }%
  }%  
  \mathrel{#3}%
  \mathrel{%   
    \mathop{\da@bars}\limits
    \ifx\\#1\\%
    \else
      _{\copy0}%
    \fi
    \ifx\\#2\\%
    \else
      ^{\copy2}%
    \fi
  }%   
  \mathrel{#4}%
}
\makeatother

\newcommand{\coloneqq}{=}

%% file: paper_sections/intro.tex
\section{Introduction}
Vector Addition Systems with States (VASS, for short)
%Petri nets 
are rich mathematical models for the description of distributed systems, as well as chemical and biological processes, and more~\cite{schmitz2016complexity}. They comprise a finite control graph with natural-valued counters, which updates the counters by adding to them integer-valued vectors that label the edges of the graph.
%They comprise a discrete control graph whose edges are labelled with integer vectors, and they operate by updating a set of natural-valued counters by traversing a path through the graph and adding the vector labels to the counters (without the counters becoming negative).
VASS arise naturally as an arguably-cleaner model than Petri nets, due to their reachability problem being polytime-interreducible with that of Petri nets.

While VASS 
%Essentially, they are graphs whose vertices are patitioned into \emph{places}
%that can hold tokens and \emph{transitions} that can be fired to move tokens
%amongst places.  While Petri nets 
are a very expressive model of
concurrency that admits algorithmic analysis, the complexity of several
associated decision problems is --- in the worst case --- prohibitively high.
For instance, the \emph{rechability problem}, which asks ``is a given
target configuration reachable from a given initial configuration?'' was recently proved to be Ackermann-complete~\cite{co21}. 
%A related model known to
%be equivalent to Petri nets in many ways is that of \emph{vector addition
%systems with states} (VASS, for short). In VASS, natural-valued counters are
%affected by adding to them integer vectors labelling the edges of a directed
%graph. In particular, the reachability problem for VASS is polynomial-time
%interreducible with the same problem for Petri nets.
%\shtodo{Is the whole introduction with Petri nets necessary? Isn't it enough to say that VASS are a central model of concurrenct, and that their reachability problems are interreducible? Then this whole discussion can have VASS in its center.}

\emph{Continuous VASS} were introduced by Blondin and
Haase~\cite{blondin2017logics} as an alternative to \emph{continuous Petri
nets}~\cite{ad10} which 
trade off the ability to encode discrete information in favor of
%trades off the inability to encode discrete information against
computational and practical benefits. A VASS essentially consists of a finite-state machine whose transitions are labelled with integer vectors. Besides the current state, a configuration of the VASS also comprises the current values of a set of counters. When a transition is taken, the state of the machine changes and the values of the counters are updated by adding to them the vector that labels the transition. A continuous VASS additionally allows to scale the update vector by some scalar $0 < \alpha \leq 1$ before adding it to the current counter values (see \autoref{fig:example-cvass}).

\begin{figure}[t]
    \centering
    \begin{tikzpicture}[initial text={},inner sep=1pt,minimum size=0mm]
        \node[state,initial left](q0){$q_0$};
        \node[state,right= of q0](q1){$q_1$};
        \node[state,right= of q1](q2){$q_2$};
        \node[state,right= of q2](q3){$q_3$};
        \path[->,auto]
        (q0) edge node[below]{$(1,0)$} (q1)
        (q0) edge[bend left] node{$(0,0)$} (q2)
        (q1) edge node[below]{$(0,1)$} (q2)
        (q2) edge[bend left] node{$(1,2)$} (q3)
        (q3) edge[bend left] node{$(2,4)$} (q2);
    \end{tikzpicture}
    \caption{A VASS of dimension $2$: From $q_0$ and with initial counter values $\vec{0}$, the state $q_2$ can be reached with counter values $\{(3i+a,6i+b) \mid (a,b) \in \{(0,0), (1,1)\}, i \in \mathbb{N} \}$; with continuous semantics, it can be reached with counter values $\vec{x} + \vec{y}$ where $\vec{x} \in \{(0,0)\} \cup \{(a,b) \mid 0 < a,b \leq 1\}$ and $\vec{y}$ lies on the ray $\{(i,2i) \mid i \in \mathbb{N}\}$}
    \label{fig:example-cvass}
\end{figure}
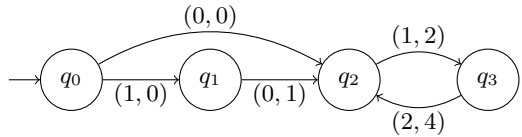

In contrast to the
situation with ``discrete'' VASS, the computational complexity of the
reachability problem for continuous VASS is rather low. Namely, in~\cite{blondin2017logics}
the reachability problem
for continuous VASS is shown to be \NP-complete while the complexity of the same problem for discrete VASS
is Ackermann-complete~\cite{co21}.

Despite the relatively low computational complexity, \NP{} is not universally
considered as tractable. 
The only subcase previously known to be in \P{} was that of cyclic reachability when counters
are allowed to hold negative values. It is also worth noting that the aforementioned \NP{} upper bound is obtained by encoding the reachability problem into
the existential fragment of the first-order theory of the reals with addition
and order. It is natural to ask whether more efficient algorithms or encodings
into ``simpler'' logics exist, e.g. linear programming, even if only for
particular subcases.

\paragraph{Reachability in fixed-dimension VASS.}
Recall that the reachability problem for discrete VASS is now known to be
Ackermann-complete. The upper bound is due to Leroux and Schmitz~\cite{ls19},
while the lower bound was independently proven by Leroux~\cite{leroux21} and
by Czerwi\'nski and Orlikowski~\cite{co21}. The latter, relatively new, lower
bound has renewed interest in what could be named the \emph{Bordeaux-Warsaw
program} --- that is, the study of the computational complexity of the
reachability problem for low-dimensional VASS and extensions thereof (see,
e.g.,~\cite{cllp19,clllm20,clo21,fls18}).  Indeed, one hopes that in such
cases, there may be efficient algorithms for the problem. Additionally, to
quote Czerwi\'nski and Orlikowski~\cite{co21}: ``it is easier to [design]
sophisticated techniques working in a simpler setting [that might] result in
finding new techniques useful in much broader generality.'' For dimensions $1$
and $2$ (and counter updates encoded in binary) the problem is known to be
\NP-complete~\cite{hkow09} and \PSPACE-complete~\cite{blondin21},
respectively.

An important structural restriction on VASS which is often used as an
intermediate step in establishing upper bounds is that of \emph{flatness},
i.e. disallowing nested cycles. In fact, the upper bounds for dimensions $1$
and $2$ mentioned above can be seen as a consequence of such VASS being
effectively flattable~\cite{ls04}. A further restriction consists of asking
that the set of all runs of the VASS can be represented by a single regular
expression $\pi_0 \chi_1^* \dots \pi_{n-1} \chi_n^* \pi_n$ over the
transitions. Such VASS are called \emph{linear path schemes} (LPS, for short).

\paragraph{VASS variants.}
In this work we study continuous VASS of fixed dimension. For complexity matters,
we assume all counter updates are encoded in binary. As decision problems, we
focus on reachability (via a path that might make the counters negative);
nonnegative reachability, i.e. reachability via a path that keeps the counters
nonnegative at all times; and zero-test reachability, corresponding to
reachability with the added constraint that some states can only be visited
with value zero for a designated counter. We summarize known and new results
regarding the complexity of these problems in \autoref{tab:summary}. Below, we give a textual account of the complexity bounds from the table.

\begin{table}[b]
\centering
\caption{Summary of computational complexity results for the reachability
problem for VASS of fixed dimension. We
write lower bounds for simpler cases and upper bounds for more general ones. New results are shown in green (upper) and red (lower bounds).}
\label{tab:summary}
\begin{tabular}{ c|c|c|c|c|c|c }
  & \multicolumn{3}{c|}{Discrete} & \multicolumn{3}{c}{Continuous} \\
  Problem & General & Flat & LPS & General & Flat & LPS \\ 
  \hline
  Reachability & in \NP & $=$ & \NP-hard & in \NP & \NP-hard & in \P\\ 
  Nonneg. reach. & in Ackermann & in \NP & \NP-hard & in \NP & \NP-hard & \textcolor{green}{in \P} \\
  Zero-test reach. & Undecidable & in \NP & \NP-hard & Undecidable & \textcolor{red}{\NP-hard} & \textcolor{red}{\NP-hard} 
\end{tabular}
\end{table}

\begin{description}
\item[(Discrete) Reachability.] The \NP-hardness bound for LPS can be shown using a simple reduction from the
\textsc{SubsetSum} problem with multiplicities, i.e. summands can be added more
than once. The latter is known to be \NP-complete, see e.g.~\cite[Proposition
4.1.1]{haase12}. The corresponding upper bound for the general case is
folklore and is proven in~\cite{hh14} even with \emph{resets}.
\item[Continuous reachability.] The \NP-hardness bound for flat VASS is stated in~\cite[Lemma 4.13(a)]{blondin2017logics} for nonnegative reachability but the reduction establishes it for reachability as well. The upper bound for the general case follows from~\cite[Corollary 4.10]{blondin2017logics}. Membership in \P{} for LPS can be derived from~\cite[Theorem 4.15]{blondin2017logics} which states that continuous cyclic reachability is in \P. In this work, we give an alternative algorithm for continuous cyclic reachability and present a full decision procedure for continuous reachability for LPS.
\item[Nonnegative reachability.] For fixed dimension $d$, no sub-Ackermannian upper bound is known. \NP-hardness for LPS follows from the same proof as for
  reachability since the construction has no negative updates. Finally, the
  \NP{} upper bound for flat VASS is folklore: (nonnegative)
  reachability in flat VASS can be encoded into existential
  Presburger-arithmetic, a theory whose decidability is \NP-complete (see, e.g.,~\cite{haase18}).
\item[Continuous nonnegative reachability.] The \NP{} upper and lower bounds for the general and flat cases follow from (the proofs of) the same results in the continuous reachability case. For the \P{} upper bound, however, one cannot rely on \cite[Theorem 4.15]{blondin2017logics}. In fact, cyclic reachability (for general dimensions) is \NP-hard in the continuous nonnegative case~\cite[Lemma 4.13(b)]{blondin2017logics}. This is, thus, the first novel complexity bound we establish.
\item[Zero-test reachability.] The \NP-hardness bound for LPS is a consequence of reachability being a subcase of zero-test reachability. The matching upper bound for flat VASS is an extension of the classical encoding\footnote{For completeness, we describe the encoding for LPS in appendix. The latter is easy to adapt for the flat case.} into PA which accounts for linear constraints imposed by the zero tests on cycles. Finally, the general model is also known as Minsky machines and its reachability problem was proven undecidable by Minsky himself~\cite{minsky67}.
\item[Continuous zero-test reachability.] The \NP-hardness for flat VASS is a
  consequence of reachability being a subcase of zero-test reachability. For
  LPS, the lower bound is novel and points to continuous zero-test
  reachability not being a suitable approximation of the discrete case. Note
  that the general case is undecidable in dimension $4$ or
  higher~\cite[Theorem 4.17]{blondin2017logics}.
\end{description}

\paragraph{Our contributions.}
Our main contribution is a geometrical understanding of the reachability sets
of continuous VASS (see \autoref{thm:cycles}, \autoref{thm:paths}, and \autoref{thm:geometry-lps}). The latter allows us (1) to prove that short LPS suffice as witnesses of (nonnegative) reachability when the dimension is fixed (see \autoref{prop:lp_scheme} and \autoref{prop:LPS-nonneg}), and (2) to give new algorithms for the reachability problem for LPS (see \autoref{thm:lpsreachp} and \autoref{thm:nonneg-lps-ptime})
via encodings of their reachability sets into tractable theories.
When possible, we stay within linear programming solutions to enable
efficient implementation of our algorithms. In some cases, we resort to
encoding problems into the Horn fragment introduced by Blondin and
Haase~\cite{blondin2017logics,blondin20}. Finally, we establish that
zero-test reachability for LPS is \NP-hard even in dimension $2$ (\autoref{thm:hard-zero}). 

For lack of space, we defer the proofs of some results to the appendix.

%% file: paper_sections/preliminaries.tex
\section{Preliminaries}
In this work, $\mathbb{Q}$ denotes the set of all rational numbers; $\mathbb{Q}_{> 0}$, all strictly positive rational numbers; and $\mathbb{Q}_{\geq 0}$, all nonnegative ones --- including $0$. Similarly, $\mathbb{N}$ denotes the set of all natural numbers including $0$.

\subsubsection{Topology and geometry.}
Let $d$ be a positive integer.
For any $\vec{x}\in \mathbb{Q}^d$, and $r\in \mathbb{R}$, we define the \emph{open ball of radius $r$ around $x$}  as usual:
\(
    B_r(\vec{x}) = \{\vec{y} \in \mathbb{Q}^d : \normtwo{\vec{x} - \vec{y}} < r\}.
\)
Let $X \subseteq \mathbb{Q}^d$. Then, we define the \emph{interior} of $X$ as 
\(
    \interior(X) = \{\vec{x}\in X \mid \exists \varepsilon > 0, B_\varepsilon(\vec{x}) \subseteq X\}.
\) 
We define the \emph{closure} of $X$ as 
\(
    \closure(X) = \{\vec{x} \in \mathbb{Q}^d \mid \forall \varepsilon > 0,\ B_\varepsilon(\vec{x}) \cap X \neq \emptyset\}
\)
and the \emph{boundary} of $X$ as 
\(
    \boundary(X) =  \closure(X) \setminus \interior(X).
\)

Let $G \subseteq \mathbb{Q}^d$ be a set of (generating) vectors. We write $\cone(G)$
to denote the \emph{(rational convex) cone} $\{\sum_{i=0}^k a_i \vec{g_i} \mid k \in
\mathbb{N}, \vec{g_i} \in G, a_i \in \mathbb{Q}_{\geq 0}\}$.  The \emph{(linear) span} of
$G$ is defined as follows: $\linspan(G) = \{\sum_{i=0}^k a_i \vec{g_i}
\mid k \in \mathbb{N}, \vec{g_i} \in G, a_i \in \mathbb{Q}\}$. Finally, the \emph{affine
hull} $\aff(G)$ of $G$ is the set $\{\sum_{i=0}^k a_i \vec{g_i} \mid k \in
\mathbb{N}, \vec{g_i} \in G, a_i \in \mathbb{Q}, \sum_{i=0}^k a_i = 1\}$. (In
particular, note that if $\vec{0} \in G$ then $\aff(G) = \linspan(G) =
\linspan(H)$, for some $H \subseteq G$ with cardinality at most $d$.) 
%\arkain{I think we can define $\linspan(G)$, $\cone(G)$ and $\aff(G)$ as
%\[
%\linspan(G) =
%\left\{\sum_{\vec{g} \in G} a_{\vec{g}} \vec{g},\
%a_{\vec{g}} \in \mathbb{Q},\
%a_{\vec{g}} \neq 0 \text{ only for finitely many }
%\vec{g} \in G
%\right\}
%\]
%\[
%\cone(G) \coloneqq
%\left\{\sum_{\vec{g} \in G} a_{\vec{g}} \vec{g},\ a_{\vec{g}} \in \mathbb{Q}_{\geq 0},\
%a_{\vec{g}} \neq 0 \text{ only for finitely many }
%\vec{g} \in G
%\right\}
%\]
%\[
%\aff(G) \coloneqq
%\left\{\sum_{\vec{g} \in G}^k a_{\vec{g}} \vec{g},\
%a_{\vec{g}} \in \mathbb{Q},\
%\sum_{\vec{g} \in G} a_{\vec{g}} = 1,\
%a_{\vec{g}} \neq 0 \text{ only for finitely many }
%\vec{g} \in G
%\right\}
%\]}

Let $X \subseteq \mathbb{Q}^d$. Then, we define the \emph{relative interior of $X$},
 $\relint(X)$, as its interior with respect to its embedding into its own affine hull --- where it
has full dimension
%\arkain{Maybe we can skip this comment about "full dimension".}
--- as follows:
\(
    \relint(X) = \{\vec{x} \in X \mid \exists \varepsilon > 0, B_\varepsilon(\vec{x}) \cap \aff(X) \subseteq X\}.
\)
%Let $G = \{\vec{g_1},\dots,\vec{g_k}\} \subseteq \mathbb{Q}^d$.
%\arkain{And then?}
%It is easy to
%see that $\relint(\cone(G)) = \{\sum_{i=0}^{k} a_i \vec{g_i} \mid a_i \in
%\mathbb{Q}_{\geq 0}, a_i > 0\}$.

\subsection{Continuous VASS}
Let $d$ be a positive integer.
A \emph{continuous VASS} $\calV$ of dimension $d$ is a tuple $(Q,T,\ell)$ where $Q$ is a
finite set of states, $T \subseteq Q \times Q$ is a finite set of transitions,
and $\ell : T \rightarrow \mathbb{Q}^d$ assigns an \emph{update label} to
every transition.

\subsubsection{Paths and runs.}
A \emph{configuration} $c \in Q \times \mathbb{Q}^d$ is a tuple consisting of
a state and the concrete values of the $d$ \emph{counters} associated with the
VASS.
We typically denote the configuration $(p,\vec{x})$ by $p(\vec{x})$.

A \emph{path} $\pi$ is a sequence $(p_1,p_2) (p_2,p_3) \dots (p_{n-1},p_{n})
\in T^*$ of transitions. We write $|\pi|$ to denote the length of the path,
i.e. $|\pi| = n - 1$. A \emph{run} $\rho$ is a sequence $q_1(\vec{x_1})
q_2(\vec{x_2}) \dots q_n(\vec{x_n})$ of configurations such
that for all $1 \leq i < n$ we have:
\begin{itemize}
    \item $(q_i,q_{i+1}) \in T$ and
    \item $\vec{x_i} + \alpha_i \cdot \ell(q_i,q_{i+1}) = \vec{x_{i+1}}$ for
      some $\alpha_i \in \mathbb{Q}$ with $0 < \alpha_i \leq 1$.
\end{itemize}
Often, we refer to the $\alpha_i$ as the \emph{coefficients of the run}.
We say $\rho$ \emph{induces} the path $(q_1,q_2) \dots (q_{n-1},q_n)$. Conversely, we
sometimes say a run is \emph{lifted} from a path. For instance, $\pi$ can be
lifted to a run $p_1(\vec{y_1}) \dots p_n(\vec{y_n})$ by fixing
$p_1(\vec{y_1})$ as initial configuration and by choosing adequate
coefficients $\alpha_i$ for all transitions.

%\shtodo{I think an example is needed here. See if you like the following.}
%\todo[inline]{G: nice, recycling the example from the intro}
As a more concrete example, consider the path $(q_0,q_1),(q_1,q_2),(q_2,q_3)$ in \autoref{fig:example-cvass}, whose transitions are labelled by $(1,0)$ and $(0,1)$. Starting from the configuration $(0,0)$ and using the coefficients $\alpha_1=0.3$ and $\alpha_2=0.5$ this path lifts to the run $q_0(0,0)q_1(0.3,0)q_2(0.3,0.5)$.

\subsubsection{Classes of continuous VASS.}
Classical VASS semantics only allows nonnegative counter values. In contrast, we consider continuous VASS both in a setting where only nonnegative counter values are allowed, denoted $\mathbb{Q}_{\geq 0}$VASS, and a setting which allows negative counters, denoted  $\mathbb{Q}$VASS.

%In contrast to the classical definition of VASS, where configurations require counter values to be nonnegative, we have defined configurations of continuous VASS to allow negative rationals. In the sequel we write $\mathbb{Q}$VASS to refer to continuous VASS as defined above and $\mathbb{Q}_{\geq 0}$VASS to refer to continuous VASS allowing only for nonnegative rational counter values.

\subsection{Reachability}
Let $p(\vec{x})$ and $q(\vec{y})$ be two configurations. We
say $q(\vec{y})$ is \emph{reachable} from $p(\vec{x})$, denoted $p(\vec{x})
\xrightarrow{*} q(\vec{y})$, if there exists a run whose first and last
configurations are $p(\vec{x})$ and $q(\vec{y})$ respectively. For a path
$\pi$, we write $p(\vec{x}) \xrightarrow{\pi} q(\vec{y})$ if, additionally,
such a run exists which can be lifted from $\pi$.

Given a configuration $p_1(\vec{x})$ and a state $q$, we define the
\emph{reachability set} of a path $\pi = (p_1,p_2)\dots
(p_{n-1},p_n)$ or a set $P$ of paths below.
\[
\textstyle
  \Reach^{\vec{x}}(\pi) = \{\vec{y} \in \mathbb{Q}^2 \mid p_1(\vec{x})
  \xrightarrow{\pi} p_n(\vec{y}) \} \quad\quad
  \Reach^{\vec{x}}(P) = \bigcup_{\pi \in P} \Reach(\pi)
\]
If $\vec{x} = \vec{0}$ then we write simply $\Reach(\pi)$ and $\Reach(P)$.

%% file: paper_sections/gengeometry.tex
\section{The Geometry of $\mathbb{Q}$VASS Reachability Sets}
\label{section:geometryRS}
In this section we discuss the geometry of the reachability sets in continuous VASS of dimension $d$. We first discuss paths and cycles separately. Then, we show that for solving the  
%use the latter to show that, to solve the 
reachability problem, we only need to take short \emph{linear path schemes} into consideration.

\subsection{The geometry of reachability sets for cycles}
For this section, we fix a cycle $\chi = (p_1,p_2) \dots (p_{m},p_{m+1})$, with $p_1 = p_{m+1}$.  We study the geometry of the set $\Reach(\chi^*)$, where $\chi^*$ stands for $\{\chi^k \mid k \in \mathbb{N}\}$.

Intuitively, following $\chi$ allows us to add a scaled version of each transition vector along $\chi$ arbitrarily many times, with the proviso that the scaling is strictly positive (the restriction to scale up to $1$ disappears since we can repeatedly take the cycle). Thus, we can intuitively reach the interior of a cone, i.e. a positive linear combination of the vectors along $\chi$. For example, a cycle with vectors $(1,0)$ and $(0,1)$ will allow us to reach $\{(x,y) \mid x>0,y>0\}\cup\{(0,0)\}$. 
However, this intuition needs to be formalized carefully to account for linear dependencies between the vectors. This may render the cone not full-dimensional, i.e. its linear span may be a strict subspace of the vector space it is in. That would mean that the interior of the cone is empty. However, in such cases, the reachability set still is a ``flattened'' version of the interior, namely the relative interior of the cone.
%(in which case the reachability set is a ``flattened'' version of the interior).
%\arkain{I think we should elaborate it slightly more.
%Namely, we should write that we call a cone $C$ not full dimensional if its linear span is a strict subspace of the vector space it is in.
%This implies that the interior of the cone is empty.
%Because if $\vec{v}$ is not spanned by $C$ then for all $\vec{u} \in C$, $\vec{u} +\lambda \vec{v} \notin C$ for all $\lambda \in \mathbb{Q}\setminus\{0\}$.
%Hence we take the relative interior of $C$ instead,
%which is defined as the interior relative to the vector space spanned by $C$, which is always non-empty and matches with our intuition.}
%\todo[inline]{G: good idea, I added most of your text above now}

\subsubsection{From cycles to cones.}\label{sec:cycles2cones}
We formalize our intuition by proving that $\Reach(\chi^*)$ is the relative interior of the cone generated by
%
%We show that the set $\Reach(\chi^*)$ has a close relation
%with %a particular cone. Namely, we refer to 
%the cone $C$ with generated by 
$G = \{\ell(p_{i},p_{i+1}) \mid 1 \leq i \leq m \}$.

%Below we formalize our intuition by proving that $\Reach(\chi*)$ is the relative interior(defined below) of $\cone(G)$.}

%The result below formalizes the intuition that the set of 

%The set of all points in $\Reach(\chi^*)$ can be obtained as a positive convex combination of elements in $G$ coincides with the interior of $\cone(G)$. 
Indeed, all points $\vec{x} \in \Reach(\chi^*)$ can be obtained as positive linear combinations of generators. To any such $\vec{x}$, we can add or subtract any generating vector and stay within $\cone(G)$, as long as it is sufficiently scaled down. Conversely, if one can add and subtract suitably scaled versions of all generating vectors to a point $\vec{x} \in \cone(G)$, and remain within $\cone(G)$, then it must be in the (relative) interior of $\cone(G)$.

Below, we write %$\interior(G)$ and
$\relint(G)$
%, and $\boundary(G)$ for $\interior(\cone(G))$, 
for $\relint(\cone(G))$.
%, and $\boundary(\cone(G))$ respectively. 
%\arkain{Maybe the following sentence is slightly better.
\begin{restatable}{theorem}{thmcycles}\label{thm:cycles}
    Let $G$ be as defined above. Then,
    \(\Reach(\chi^*) =
    \relint(G) \cup \{\vec{0}\}.\)
\end{restatable}
Note that we consider the set $G$ of all labels of transitions from $\chi$,
ignoring the fact that multiple transitions can have the same label. This is justified by the following lemma.

\begin{restatable}{lemma}{simplebasis}\label{lem:simple-basis}
    Let $\vec{x} \in \mathbb{Q}^d$. Then, there exists $(a_1,\dots,a_n) \in
    \mathbb{Q}_{>0}^n$ such that $\vec{x} = \sum_{i=1}^n a_i \vec{g_i}$ if and
    only if $\vec{x} \in
    \Reach(\chi^k)$ for some $k \in \mathbb{N}$ with $k \geq 1$.
\end{restatable}
%
%\arkain{Maybe we should mention that the proof is in the appendix.}
%
\subsubsection{A concrete case: dimension 2.}
For intuition, we state a consequence of \autoref{thm:cycles} in dimension $2$. For $d = 2$, a cone $C$ is \emph{trivial} if there exists $\vec{v} \in \mathbb{Q}^d$ such that $C$ is a subset of the line $\{r \cdot \vec{v}\ \mid\ r \in \mathbb{Q}\}$, and otherwise it is full-dimensional. For a trivial cone, its relative interior is either the entire cone (if it is the entire line), or the cone without $\vec{0}$, if it is a ray. In either case, $\relint(G)\cup\{\vec{0}\}=\cone(G)$, and it is easy to see that for a cycle $\chi$ whose vector labellings $G$ are co-linear to $\vec{v}$, the reachability set of $\chi^*$ in the continuous semantics is $\cone(G)$ (where $\vec{0}$ is added by taking $\chi$ zero times).
For full dimensional cones, we can take either $\vec{0}$ or any any positive combination of the generators, but since no element of the generators can be taken zero times, the reachability set excludes the boundary (see \autoref{fig:cone_zono_example} (left) for a visualization). We thus have the following simple statement, where we write $\interior(G)$ and
$\boundary(G)$ for $\interior(\cone(G))$ and
$\boundary(\cone(G))$ respectively.

\begin{corollary}\label{cor:cycles}
    Let $G$ be as above. In dimension $d=2$, one of the following holds.
    \begin{itemize}
        \item Either $\cone(G)$ is trivial and $\Reach(\chi^*) =
          \cone(G)$; \label{test}
        \item or $\cone(G)$ is nontrivial and $\Reach(\chi^*) \setminus \{\vec{0}\}
           = \interior(G) \setminus \boundary(G)$.
    \end{itemize}
\end{corollary}
%\arkain{I think we should mention that the corollary is specifically for dimension 2.}

\begin{figure}[t]
    \centering
    \begin{tikzpicture}[scale=0.5]
        
        \draw[fill=orange!30, color=orange!30] (0, 0)--(2, 5)--(5,5)--(5, 3)--(0, 0);
        \draw[help lines, color=gray!30, dashed] (-1.9,-1.9) grid (4.9,4.9);
        \draw[->, thick] (-2,0)--(5,0) node[right]{$x$};
        \draw[->, thick] (0,-2)--(0,5) node[above]{$y$};
        
        \draw[ultra thick,color=blue] (0, 0)--(2, 5);
        \draw[ultra thick,color=blue] (0, 0)--(5,3);
        
    \end{tikzpicture}
    \qquad \qquad
    \begin{tikzpicture}[scale=0.5]
        
        \draw[fill=orange!30] (3, 0)--(1.5, -1)--(0,0)--(0, 2)--(1.5, 3)--(3,2)--(3,0);
        \draw[help lines, color=gray!30, dashed] (-1.9,-1.9) grid (4.9,4.9);
        \draw[->, thick] (-2,0)--(5,0) node[right]{$x$};
        \draw[->, thick] (0,-2)--(0,5) node[above]{$y$};
        
        \draw[ultra thick,color=orange] (3, 0)--(1.5, -1)--(0,0)--(0, 2)--(1.5, 3);
        \draw[ultra thick,color=blue] (3, 0)--(3,2)--(1.5, 3);
        \draw[blue,fill=white] (3,0) circle (3pt);
        \draw[blue,fill=white] (1.5,3) circle (3pt);
        
        \node[] at (5, 2)  (c)     {$\sigma_G = (3,2)$};
        
    \end{tikzpicture}
    \caption{On the left: a cone with its boundary in blue; on the right: a zonotope with $G=\{(1.5, 1), (1.5, -1), (0, 2)\}$, where $\adj(G)$ is drawn in blue}
    \label{fig:cone_zono_example}
\end{figure}
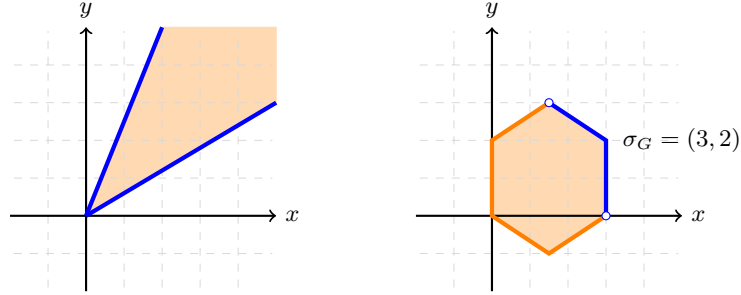

\subsection{The geometry of reachability sets for paths}

For this section, we fix a path 
$\pi = (p_1, p_2) \dots (p_m, p_{m+1})$.
Similarly to the case of cycles, we establish a connection
between $\Reach(\pi)$ and a particular type of polytope known as a
zonotope. 

Intuitively, since we now have a path that is taken once (rather than a cycle), the restriction that the scaling is at most $1$ comes into play, and limits us. Moreover, multiplicities of linearly dependent vectors along the path also matter. This renders the reachability set somewhat more complicated.

\subsubsection{Zonotopes.}
Let $G = \{\vec{g_1},\dots,\vec{g_k}\} \subseteq \mathbb{Q}^d$ be a finite set of
(generating) vectors. We write $\zono(G)$ to denote the \emph{zonotope}\footnote{Zonotope is the standard term, but since we do not use any of its special properties, the reader may view this as a standard polytope.}
%\shtodo{How about a footnote here that Zonotopes are the standard name for this object, but we will not actually use any specific properties of them? I think it would calm the reader.}
%\todo[inline]{G: arka and tim will strongly agree with you, they thought I made up the word}
$\{\sum_{j=1}^k a_i \vec{g_i} \mid a_i \in \mathbb{Q}, 0 \leq a_i \leq 1\}$.

Following our interior-based approach for cycles, we study the reachable part of the boundary of a zonotope. We define $\sigma_G$ as the sum of all vectors in $G$, or $\sigma_G \coloneqq \sum_{\vec{g}\in G} \vec{g}$.
Additionally, we define the set $\adj(Z)$ of \emph{faces} of $\zono(G)$ that are \emph{adjacent to $\sigma_G$} below. A face of $\zono(G)$ is any non-empty intersection of $\zono(G)$ with a half-space $H$ such that none of the interior points of $\zono(G)$ lies on the boundary of $H$.
\begin{definition}[Adjacent Sets]
We define $\adj(G)$ as the union of $\{\sigma_G\}$ and all $\vec{x} \in \mathbb{Q}^d$ on the relative interior of a face of $\zono(G)$ that contains $\sigma_G$.
\end{definition}
Observe that $\adj(G) = \emptyset$ whenever $\sigma_G \in \interior(\zono(G))$. In \autoref{fig:cone_zono_example} we illustrate this definition.

\subsubsection{From paths to zonotopes.}\label{sec:paths2zonotopes} 
We show $\Reach(\pi)$ has a close relation with a zonotope whose generator is derived from the \emph{multiset} $M \coloneqq [\ell(p_i,p_{i+1}) \mid 1
\leq i \leq m]$. 
Intuitively, we obtain from $M$ a generator $G_M$ by summing together co-linear vectors that are in the same ``orientation''. For example, the vectors $(1,0)$ and $(2,0)$ along a single path have the same effect as $(3,0)$, but $(1,0)$ and $(-1,0)$ have distinct effects. Technically, this is captured by grouping together vectors that are in the cones of each other. More formally, 
choose $M' \subseteq M \setminus \{\vec{0}\}$ such that for all $\vec{u} \in M \setminus \{\vec{0}\}$ there is a unique vector $\vec{u}' \in M'$ such that $\vec{u} \in \cone{(\vec{u}')}$. Then,
define $G_M$ as follows:
\(
G_M \coloneqq
\left\{ \sum_{\vec{u} \in M \cap \cone{(\vec{u}'})} \vec{u}
\:\middle|\:
\vec{u}' \in M'
\right\}.
\)

We show that the reachability set of $\pi$ can be computed by taking $\zono(G_M)$, and removing its boundary except for faces adjacent to $\sigma_G$. The intuition behind this is similar to that of cycles: one can add strictly positive-scaled versions of the generating vectors, and therefore boundary elements that are obtained using 0-scales are unreachable. However, in zonotopes there are additional boundary faces that are obtained by capping the scale at $1$ on some elements, and these are the faces adjacent to $\sigma_G$ (with $\sigma_G$ itself being the reachable vector where all elements are scaled to $1$).

Below, we write $\relint(G_M)$ to
denote $\relint(\zono(G_M))$.
\begin{restatable}{theorem}{pathstozonotopes}\label{thm:paths}
  Let $G_M$ be as above. Then,
  %we have that:
  \(
    \Reach(\pi) = \relint(G_M) \cup \adj(G_M).
  \)
\end{restatable}
%For the rest of this section, we write $G = \{\vec{g_1},\dots,\vec{g_n}\}$ for the generator $G_M$.

As with cycles, compared to the multiset of all path labels, we restrict our attention to a simpler set $G$ of vectors. The following result connecting them
is almost immediate from the definitions.
\begin{restatable}{lemma}{simplebasispath}\label{lem:simple-basis-path}
  Let $\vec{x} \in \mathbb{Q}^d$. Then, there exists $(a_1,\dots,a_n) \in
  \mathbb{Q}^n$ such that $\vec{x} = \sum_{i=1}^n a_i \vec{g_i}$ and $0 < a_i
  \leq 1$, for all $1 \leq i \leq n$, if and only if $\vec{x} \in
  \Reach(\pi)$.
\end{restatable}

\subsection{The geometry of reachability sets for linear path schemes}\label{sec:lps}
We now turn to combine our geometrical understanding of cycles (\autoref{sec:cycles2cones}) and paths (\autoref{sec:paths2zonotopes})
%In the previous sections, we studied the geometry of the reachability sets of paths and cycles (where we can follow a cycle arbitrarily many times). Now, we will use that 
to describe the geometry of a \emph{linear path scheme}. 

A linear path scheme (LPS, for short) is a regular expression  $\pi_0 \chi_1^* \pi_1 \dots
\chi_n^* \pi_n$ over the transitions. Importantly, the $\pi_i$ are
(possibly empty) paths; the $\chi_i$ are cycles; and $\pi_0 \chi_1 \pi_1 \dots
\chi_n \pi_n$ is a valid path. Each LPS determines an infinite
set $\{\pi_0 \chi_1^{k_1} \dots \chi_n^{k_n} \pi_n \mid k_1, \dots, k_n \in
\mathbb{N}\}$ of paths that follows each of the paths exactly once and each of
the cycles an arbitrary number of times.

For the rest of this section, we fix a
linear path scheme $\sigma = \pi_0
\chi_1^* \pi_1 \dots \chi_n^* \pi_n$.

\subsubsection{From LPS to cones and a zonotope.}
From previous developments in this work, the reader might already believe that the reachability set of an LPS can be shown to be the Minkowski sum of suitable subsets of convex cones and zonotopes. It transpires that one can further simplify this and obtain a characterization which involves a single zonotope and a polynomial number of cones. 

Below, we write $G(\pi)$ to denote the generator of the zonotope for the path $\pi_0 \pi_1 \dots \pi_n$ as defined in \autoref{sec:paths2zonotopes}; for each $1 \leq i \leq n$,
we write $G(\chi_i)$ to denote the generator of the convex cone for the cycle $\chi_i$ as defined in \autoref{sec:cycles2cones}.

\begin{restatable}{theorem}{geometrylps}\label{thm:geometry-lps}
    Let $I_1,\dots,I_k \subseteq
    \{1,2,\dots,n\}$ be a partition such
    that if $i,j \in I_\ell$ for some $\ell$ then $\linspan(G(\chi_i))=\linspan(G(\chi_j))$. Then, $\Reach(\sigma)$ is:\footnote{To avoid clutter, we omit some parentheses: we assume union has higher precedence than Minkowski sum.}
    \[
    \textstyle
    \relint(\zono(G(\pi))) \cup \adj(G(\pi))
        + \sum_{j=1}^k \relint\left(\cone\left(\bigcup_{i \in I_j} G(\chi_i)\right)\right) \cup\{\vec{0}\}.
    \]
\end{restatable}

To prove the theorem, we establish two intermediate results. The first one, together with \autoref{thm:paths}, already yields the first (Minkowski) summand from the expression in \autoref{thm:geometry-lps}. The result below follows immediately from the definitions and commutativity of the Minkowski sum.

\begin{restatable}{lemma}{reachissum}\label{lem:reachissum}
    We have that
    \(
        \Reach(\sigma) = \Reach(\pi_0 \pi_1 \dots \pi_n) + \sum_{i=1}^n\Reach(\chi^*_i).
    \)
\end{restatable}

The next result allows us to group the sums of cycle rechability sets into convex cones for each common linear subspace. Indeed, \autoref{thm:cycles} and an induction on the following lemma yield the last summands from \autoref{thm:geometry-lps}.

\begin{restatable}{lemma}{sumcones}\label{lemma:sumcones}
    Let $C$ and $C'$ be cones with generators $G$ and $G'$ respectively. Then, \(C + C' = \cone(G \cup G')\) and \(\relint(C + C') = \relint(C) + \relint(C')\). If, additionally, $\linspan(G) = \linspan(G')$, then
    \(
    \relint(C), \relint(C') \subseteq
    \relint(C + C').
    \)
\end{restatable}

%% file: paper_sections/QReach.tex
\section{The Complexity of $\mathbb{Q}$VASS Reachability}\label{sec:qvass-reach}
In this section, we use our results concerning the geometry of reachability sets to give an \NP{} decision procedure for the reachability problem.

\begin{theorem}
\label{thm:QVASS_reach_inNP}
    For every $d \in \mathbb{N}, d \geq 1$, given a $\mathbb{Q}$VASS of dimension $d$, and two configurations $p(\vec{x})$ and $q(\vec{y})$, determining whether $p(\vec{x}) \xrightarrow{*} q(\vec{y})$ is in \NP.
\end{theorem}

 First, without loss of generality, we assume $\vec{x} = \vec{0}$. Indeed, $p(\vec{x}) \xrightarrow{*} q(\vec{y})$ if and only if $p(\vec{0}) \xrightarrow{*} q(\vec{y}-\vec{x})$.
In the following we prove
that if $p(\vec{0}) \xrightarrow{*} q(\vec{x})$ then there is a linear path scheme $\sigma$ such that $p(\vec{0}) \xrightarrow{\pi} q(\vec{x})$ for some $\pi \in \sigma$ and $\sigma$ has size polynomial on the number of transitions $|T|$ and exponential on the dimension $d$. Then, we show that checking whether $p(\vec{0}) \xrightarrow{*} q(\vec{x})$ under a given linear path scheme is decidable in polynomial time. It follows that to check whether a configuration is reachable, in a general $\mathbb{Q}$VASS, one can guess a polynomial-sized linear path scheme and check (in polynomial time) whether the corresponding configuration is reachable in it.

\subsection{Short linear path schemes suffice}

Presently, we argue that for any path we can find an LPS with a number of cycles that is polynomial in the number of transitions of the $\mathbb{Q}$VASS and exponential on the dimension such that all paths and cycles are simple, the set of transitions in the LPS is the same as that in the path, and the reachability set of the LPS includes that of the path. 

For convenience, we define the \emph{support} of a path $\pi = t_1 \dots t_n$ as the set of all transitions that are present in the path: $\llbracket \pi \rrbracket \coloneqq \{ t_i \mid 1 \leq i \leq n\}$.
For an LPS $\sigma = \pi_0 \chi_1^* \dots \chi^*_n \pi_n$, its support is $\llbracket \sigma \rrbracket \coloneqq \llbracket \pi_0 \rrbracket \cup \bigcup_{i=1} \llbracket \chi_i \rrbracket \cup \llbracket \pi_i \rrbracket$.

\begin{theorem} \label{prop:lp_scheme}
    Let $\pi$ be an arbitrary path. Then, there exists a linear path scheme $\sigma =  \pi_0 \chi_1^* \pi_1 \dots \chi_n^* \pi_n$ such that:
    \begin{itemize}
        \item $n \leq d|T|^{d+1}$,
        \item $\pi_i$ is a simple path for all $0 \leq i \leq n$,
        \item $\chi_i$ is a simple cycle for all $1 \leq i \leq n$,
        \item $\llbracket \pi \rrbracket = \llbracket \sigma \rrbracket$, and
        \item $\Reach(\pi) \subseteq \Reach(\sigma)$.
    \end{itemize}
\end{theorem}

Most properties of the LPS in the result above follow from considering an LPS with minimal length, with the length of an LPS defined as $|\sigma| \coloneqq |\pi_0| + \sum_{i=1}^n |\chi_i| + |\pi_i|$. The main technical hurdle is thus the upper bound on the number of cycles. Our approach is to remove cycles whose support is covered by other cycles. The result below, which follows directly from \autoref{thm:cycles} and \autoref{lemma:sumcones}, gives us that flexibility. As in \autoref{sec:lps}, we write $G(\chi)$ to denote the generator of the convex cone for $\chi$, i.e. $G(\chi) = \{\ell(t) \mid t \in \llbracket \chi \rrbracket\}$.

\begin{restatable}{lemma}{cyclerep}\label{lemma:cyclerepresentation}
    Let $\chi$ be a cycle and $C$ be a set of cycles. If $\llbracket \chi \rrbracket \subseteq \bigcup_{\theta \in C} \llbracket \theta \rrbracket$ and $\linspan(G(\theta_1)) = \linspan(G(\theta_2))$ for all $\theta_1,\theta_2 \in C \cup\{\chi\}$ then: 
    \[
    \textstyle
    \Reach(\chi^*) + \sum_{\theta \in C} \Reach(\theta^*) = \sum_{\theta \in C} \Reach(\theta^*).
    \]
\end{restatable}
%
%\arkain{I think it's enough to have
%\[
%\linspan(G(\theta)) \subseteq \linspan(G(\chi))
%\text{ for all }\theta \in C \ .
%\]
%But for our purposes, probably this is irrelevant.
%}
%\arkain{Where is the proof?}
%\todo[inline]{G: we write above that it follows directly from theorem 1 and lemma 9}
%
\begin{proof}[of \autoref{prop:lp_scheme}]
By the pigeonhole principle, if $D$ is a set of cycles with $|D| > |T|$ then
there is some $\chi \in D$ such that $\llbracket \chi \rrbracket \subseteq
\bigcup_{\theta \in C} \llbracket \theta \rrbracket$, with $C \coloneqq D
\setminus \{\chi\}$. It thus follows from \autoref{lemma:cyclerepresentation}
that to bound the number of cycles from the LPS in \autoref{prop:lp_scheme} it
suffices to bound the size of the set of possible spans. Since the labels of
transitions come from $\mathbb{Q}^d$, the dimension of any span $V$ is at most
$d$. Now, every $V$ is generated by at most $d$ linearly independent
transition labels, i.e. $d$ rational vectors.
Hence, the number of possible $V$ is at most $\sum_{i = 1}^d{{|T|} \choose i} \leq d|T|^d$.
For every $V$, we need at most $|T|$ cycles so that their support covers that of any superset generating the same span, hence the bound $d|T|^{d+1}$ in \autoref{prop:lp_scheme}.
\end{proof}

From the results above, it follows that to check whether a configuration is
reachable in a general $\mathbb{Q}$VASS of fixed dimension, one can guess a
polynomial-sized linear path scheme and check whether the corresponding
configuration is reachable in it. In order to conclude membership in \NP{},
it remains to argue that the latter check can be realized in polynomial time.

\subsection{Reachability for linear path schemes is tractable}
In this section, we show that determining whether a configuration is reachable via a linear path scheme is decidable in polynomial time.

\begin{theorem}\label{thm:lpsreachp}
    Given a linear path scheme $\sigma$ and $\vec{x},\vec{y}\in \mathbb{Q}^d$,
    determining whether $\vec{y} \in \Reach^{\vec{x}}(\sigma)$ is in \P.
\end{theorem}

Based on \autoref{lem:reachissum} and the geometric
characterizations of the reachability sets of cycles and paths, it suffices to
show how to determine whether there exist $\vec{z},\vec{c_1},\vec{c_2},\ldots,\vec{c_n}
\in \mathbb{Q}^d$ in a zonotope and $n$ cones, respectively, such
that $\vec{y} = \vec{z} + \sum_{i=1}^n \vec{c_i}$. To do so, we make
\autoref{lem:simple-basis-path} and \autoref{lem:simple-basis} effective by
encoding the former into a system of linear inequalities with
\emph{strict-inequality constraints} and the latter as a conjunction of \emph{convex semilinear Horn clauses}. 

It is known that the feasibility problem for linear programs is decidable in polynomial time even in the presence of strict inequalities (see, e.g.,~\cite[Ch. 8.7.1]{ps82}). A convex semilinear Horn clause (CSH clause, for short) is of the form:
\[
    \textstyle
    (\vec{a} \cdot \vec{x} \prec b) \lor \bigvee_{i=1}^m \bigwedge_{j \in J_i} x_j > 0, 
\]
where $\vec{a} \in \mathbb{Q}^d, b \in \mathbb{Q}$ and ${\prec} \in \{=,<,\leq\}$ and $J_1,\ldots,J_m$ are sets of indices. 
Crucially, the feasibility problem for systems of CSH clauses over \emph{non-negative} variables $\vec{x}\in \mathbb{Q}_{\ge 0}^d$ is known to be decidable in polynomial time~\cite{blondin2017logics,blondin20}. As we show in the sequel, describing membership in a cone using a CSH is straightforward over $\mathbb{Q}$, but in order to obtain a CSH over $\mathbb{Q}_{\ge 0}$ we need some manipulations.

Henceforth, we fix an LPS $\sigma = \pi_0 \chi_1^* \pi_1 \dots
\chi_n^* \pi_n$. We also adopt the notation from \autoref{sec:lps}: $G(\pi)$
denotes the generator of the zonotope for the path $\pi_0 \pi_1 \dots \pi_n$; and $G(\chi_i)$ the generator of the cone for $\chi_i$ for each $1 \leq i \leq n$.

\subsubsection{Encoding the zonotope.}
Let $G(\pi) = \{\vec{g_1},\dots, \vec{g_m}\}$.
We now define the matrix $\vec{A} \in \mathbb{Q}^{d \times (m + d)}$ and
the vector $\vec{a} \in \mathbb{Q}^d$ as follows:
\begin{equation}\label{eqn:a-paths}
    \vec{A} = 
    \begin{pmatrix}
        (\vec{g_1})_1 & (\vec{g_2})_1 & \dots & (\vec{g_m})_1 & -1 & 0     & \dots & 0\\
        (\vec{g_1})_2 & (\vec{g_2})_2 & \dots & (\vec{g_m})_2 & 0  & -1    & 0     & \dots \\
        \vdots & \vdots & \vdots & \vdots & \vdots & \vdots & \ddots & \vdots \\
        (\vec{g_1})_d & (\vec{g_2})_d & \dots & (\vec{g_m})_d & 0  & \dots & 0     & -1
    \end{pmatrix}
    \text{ and }
    \vec{a} = (0, \dots, 0),
\end{equation}
Further, we define the matrix $\vec{B} \in \mathbb{Q}^{m \times (m + d)}$ and the vector $\vec{b} \in \mathbb{Q}^{m}$ as:
\begin{equation}\label{eqn:b-paths}
    \vec{B} = 
    \begin{pmatrix}
        \vec{I} & \vec{0} & \dots & \vec{0}
    \end{pmatrix}
    \textrm{ and }
    \vec{b} = (1, 1, \dots, 1),
\end{equation}
where $\vec{I}$ is the $m \times m$ identity matrix and $\vec{0}$ is the $m \times 1$ zero vector. Finally, we define $\vec{C} \in \mathbb{Q}^{m \times (m + d)}$ and $\vec{c} \in \mathbb{Q}^m$ as follows.
\begin{equation}\label{eqn:c-paths}
    \vec{C} = 
    \begin{pmatrix}
       -\vec{I} & \vec{0} & \dots & \vec{0}
    \end{pmatrix}
    \text{ and }
    \vec{c} = (0, 0, \dots, 0)
\end{equation}

\begin{lemma}\label{lem:enc-path}
    The system $\vec{A} \vec{z} = \vec{a} \land \vec{B}\vec{z} \leq \vec{b}
    \land \vec{C}\vec{z} < \vec{c}$ has a solution $(\vec{\alpha}, \vec{y}) \in
    \mathbb{Q}^{m + d}$ if and only if $\vec{y} \in \relint(\zono(G(\pi))) \cup \adj(G(\pi))$.
\end{lemma}
This follows immediately from \autoref{thm:paths}, \autoref{lem:simple-basis-path}, and the fact that, by construction, the system has a solution if and only if there exists $(\alpha_1,\dots,\alpha_m) \in \mathbb{Q}^m_{> 0}$ such that $\vec{y} = \sum_{i=1}^m \alpha_i \vec{g_i}$.

\subsubsection{Encoding any cone.}
Let $1 \leq i \leq n$ and $G(\chi_i) = \{\vec{g_1},\dots,\vec{g_m}\}$. We define the matrix $\vec{A} \in \mathbb{Q}^{d \times (m+d)}$ and vector $\vec{a}\in \mathbb{Q}^d$ as in \autoref{eqn:a-paths}; and $\vec{C} \in \mathbb{Q}^{m \times (m+d)}$ and $\vec{c} \in \mathbb{Q}^d$ as in \autoref{eqn:c-paths}. Consider the following system of CSH clauses:
\begin{equation}
\label{eq:CSH_cones}
\textstyle
    (\vec{A} \vec{z} = \vec{a}) \land
    \left((\vec{C}\vec{z} = \vec{c}) \lor \bigwedge_{i=1}^m \vec{z_i} > 0\right) 
\end{equation}
   
It is easy to see that the system has a solution $(\vec{\alpha}, \vec{y}) \in \mathbb{Q}^{m + d}$ if and only if $\vec{y} \in \relint(\cone(G(\chi_i))) \cup \{0\}$, $\vec{\alpha}\in \mathbb{Q}^m$ is the coefficient vector for the cone generator, and $\vec{y}$ is the target vector (cf. \autoref{thm:cycles} and \autoref{lem:simple-basis}). Unfortunately, such a solution need not have $\vec{y}\in \mathbb{Q}_{\ge 0}^d$, so we cannot use the polynomial-time algorithm of \cite{blondin2017logics,blondin20}, which requires quantification over nonnegatives.
Note that we can assume $\vec{\alpha}\in \mathbb{Q}_{\ge 0}^m$, as this is enforced by the system.

To circumvent this, we encode $\vec{y}$ as a difference vector $\vec{y}=\vec{r}-\vec{s}$, so that $\vec{r},\vec{s}\in \mathbb{Q}_{\ge 0}^m$, as follows. We modify the matrix $\vec{A}$ to a matrix $\vec{A'} = 
    \begin{pmatrix}
       \vec{A} & \vec{I}
    \end{pmatrix}\in \mathbb{Q}^{d\times (m+2d)}$
where $\vec{I}$ is the $d\times d$ identity matrix. We also define $\vec{C'}=\begin{pmatrix}
       \vec{C} & \vec{0}\ldots \vec{0}
    \end{pmatrix} \in \mathbb{Q}^{d\times (m+2d)}$.
\begin{lemma}\label{lem:enc-cones}
    The following system of CSH clauses:
    \begin{equation}
    \label{eq:CSH_cones_pos}
    \textstyle
    (\vec{A'} \vec{z} = \vec{a}) \land
    \left((\vec{C'}\vec{z} = \vec{c}) \lor \bigwedge_{i=1}^m \vec{z_i} > 0\right)
    \end{equation}
    has a solution $(\vec{\alpha},\vec{r},\vec{t})\in \mathbb{Q}_{\ge 0}^{m+2d}$ if and only if $\vec{r}-\vec{t}\in \relint(\cone(G(\chi_i))) \cup \{\vec0\}$.
\end{lemma}
The lemma follows by observing that every solution $(\vec{\alpha},\vec{y})$ for \autoref{eq:CSH_cones} induces a solution $(\vec{\alpha},\vec{r},\vec{t})$ for \autoref{eq:CSH_cones_pos} by setting e.g. $\vec{r}=\max\{\vec{y},\vec{0}\}$ and $\vec{t}=-\min\{\vec{y},\vec{0}\}$. Conversely, every solution $(\vec{\alpha},\vec{r},\vec{t})$ induces the solution $(\vec{\alpha},\vec{r}-\vec{t})$.

\begin{proof}[of \autoref{thm:lpsreachp}]
The result follows from the fact that \autoref{lem:reachissum} can be made effective by encoding it into a master system of CSH clauses conjoining the subsystems for the zonotope and all cycles, and using the fact that solutions for the system in~\autoref{lem:enc-path} can be assumed to be positive using the same methods by which we obtained \autoref{lem:enc-cones}.
\end{proof}

To conclude this section we remark that for the case of fixed dimension $d=2$ there is a simpler algorithm than the one described above.
\begin{remark}[Linear inequalities suffice in dimension $2$]
    It follows from \autoref{cor:cycles} that we can classify cones as trivial and nontrivial. Observe that the Minkowski sum of two trivial cones is trivial, and similarly for nontrivial cones. This means that \autoref{thm:geometry-lps} gives us the reachability set as a sum of a zonotope and two cones. Using our encoding of a zonotope from \autoref{lem:enc-path}, that of \autoref{eq:CSH_cones} for nontrivial cones, and its natural relaxation for trivial cones, we obtain a disjunction of a constant number of linear programs. 
    %Thus, we have reduced reachability to a constant number of linear-program feasibility subproblems. 
    This yields a simpler polynomial-time algorithm as it avoids the use of systems of CSH clauses.
\end{remark}

%% file: paper_sections/QposReach.tex
\section{The Complexity of $\mathbb{Q}_{\geq 0}$VASS Reachability}\label{sec:qposreach}
We now use results from previous sections to give an \NP{} decision procedure for the reachability problem for $\mathbb{Q}_{\geq 0}$VASS.
\begin{restatable}{theorem}{pronpposreach}
    For every $d \in \mathbb{N}, d \geq 1$, given a $\mathbb{Q}_{\geq 0}$VASS of dimension $d$, and configurations $p(\vec{x})$ and $q(\vec{y})$, determining whether $p(\vec{x}) \xrightarrow{*} q(\vec{y})$ is in \NP.
\end{restatable}
%\shtodo{I would say theorem again}
%\todo[inline]{G: sure, we can stick to theorems and lemmas, makes it cleaner}
The structure of our argument is similar to the one presented in \autoref{sec:qvass-reach}: we first prove that short LPS suffice and then we prove reachability is tractable for LPS. 
The first part is considerably more complicated for $\mathbb{Q}_{\geq 0}$VASS and it is based on the fact that short LPS exist for $\mathbb{Q}$VASS. For this reason, we need additional notation: 
We write $p(\vec{x}) \xrightarrow{*} q(\vec{y})$ to denote that $q(\vec{y})$ is reachable from $p(\vec{x})$ in a $\mathbb{Q}_{\geq 0}$VASS and instead use $p(\vec{x}) \xdashrightarrow{*} q(\vec{y})$ to denote that $q(\vec{y})$ is reachable from $p(\vec{x})$ with respect to $\mathbb{Q}$VASS semantics (i.e. when allowing negative counter values). Similarly, we write $\Reach^{\vec{x}}(\cdot)$ for reachability sets w.r.t. $\mathbb{Q}$VASS and $\Reach^{\vec{x}}_{\geq 0}(\cdot)$ for that w.r.t. $\mathbb{Q}_{\geq 0}$VASS.

We make repeated use of the following result by Blondin and Haase.
\begin{lemma}[From {\cite[Proposition 4.5]{blondin2017logics}}]\label{lemma45}
    There exists a path $\pi$ such that $q(\vec{x}) \xrightarrow{\pi}
    q(\vec{y})$ if and only if there exist paths $\pi_1,\pi_2,\pi_3$ such
    that:
    \begin{enumerate}
      \item $q(\vec{x}) \xdashrightarrow{\pi_2} q(\vec{y})$; \label{itm:45itm1}
        \item $q(\vec{x}) \xrightarrow{\pi_1} q(\vec{x'})$, for some $\vec{x'} \in \mathbb{Q}_{\geq 0}^d$; \label{itm:45itm2}
        \item $q(\vec{y'}) \xrightarrow{\pi_3} q(\vec{y})$, for some $\vec{y'} \in \mathbb{Q}_{\geq 0}^d$; and \label{itm:45itm3}
        \item $\llbracket \pi \rrbracket = \llbracket \pi_1 \rrbracket =
          \llbracket \pi_2 \rrbracket = \llbracket \pi_3 \rrbracket$.
          \label{itm:45itm4}
    \end{enumerate}
\end{lemma}
Intuitively, this means that $q(\vec{x}) \xrightarrow{\pi}
    q(\vec{y})$ if the following conditions hold: first, we obviously need $q(\vec{x}) \xdashrightarrow{\pi_2} q(\vec{y})$, and second, we need $q(\vec{x})$ to allow some ``wiggle room'' using the same transitions as $\pi$ and while keeping the counters nonnegative. Similarly, there should be wiggle room to reach $q(\vec{y})$ with nonnegative counters. %In fact, the proof of \autoref{lemma45} in
%\cite{blondin2017logics} constructs a run lifted from a path in $\pi_1 \pi_2^* \pi_3$. 
The lemma also shows that these conditions are necessary.

\subsection{Short linear path schemes suffice}
We start by proving that short LPS suffice. 
%Concretely, we establish
%the following.
\begin{theorem} \label{prop:LPS-nonneg}
  Let $\pi$ be an arbitrary path such that $p(\vec{x}) \xrightarrow{\pi}q(\vec{y})$.
  Then, there exists a linear path scheme
  $\sigma =  \pi_0 \chi_1^* \pi_1 \dots \chi_n^* \pi_n$ such that:
  \begin{itemize}
    \item $n \leq |Q|$,
    \item $\pi_i$ is a simple path for all $0 \leq i \leq n$,
    \item $|\chi_i| \leq 4|Q|(|T| + d + 2)(d|T|^{d+1}+1)$ for all $1 \leq i \leq n$, and
    \item $\vec{y} \in \Reach^{\vec{x}}_{\geq 0}(\sigma)$.
  \end{itemize}
\end{theorem}

The crux of our strategy lies
in getting an effective version of the ``only if'' direction of
\autoref{lemma45}. 
We establish intermediate results towards proving that $\pi_1$, $\pi_2^*$, and $\pi_3$ all have short LPS and that, together, they form a short
LPS as desired.
To begin, we argue that $\pi_1$ and $\pi_3$ admit short LPS.

\begin{lemma} \label{lem:LPS-cover}
  Let $\pi$ be an arbitrary path such that $p(\vec{x}) \xrightarrow{\pi}
  q(\vec{y})$. Then, there exists a linear path scheme
  $\sigma =  \pi_0 \chi_1^* \pi_1 \dots \chi_n^* \pi_n$ such that:
  \begin{itemize}
    \item $n \leq |T| + d$,
    \item $\pi_i$ is a simple path for all $0 \leq i \leq n$,
    \item $\chi_i$ is a simple cycle for all $1 \leq i \leq n$, and
    \item $p(\vec{x}) \xrightarrow{\tau} q(\vec{y'})$ for some
      $\vec{y'} \in \mathbb{Q}^d_{\geq 0}$ and $\tau \in \sigma$ such that
      $\llbracket \pi \rrbracket = \llbracket \tau \rrbracket$.
  \end{itemize}
\end{lemma}

\begin{proof}
To prove the lemma, observe that for every path $\pi$ we can obtain an LPS $\sigma$ that satisfies all but the length requirement. Indeed, any decomposition of $\pi$ into simple paths and simple cycles satisfies the second and third constraints, and since $p(\vec{x}) \xrightarrow{\pi} q(\vec{y})$, then setting $\vec{y'}=\vec{y}$ and $\tau=\pi$ gives the third constraint. 
Now, assume $\sigma = \pi_0 \chi_1^* \dots \chi_n^* \pi_n$ is an LPS satisfying all but the length constraint, whose length $n$ is minimal. We show that $n\le |T|+d$. Consider a concrete run $\rho$ lifted from a path in $\sigma$, starting from $q_0(\vec{x})$. For every $1\le j\le d$, denote by $i_j$ the index of the transition in $\rho$ after which counter $j$ becomes strictly positive (or $\infty$ if there is no such prefix). E.g., if $\vec{x}=(4,0)$, and counter 2 becomes positive after the prefix of $\rho$ corresponding to $(q_0,q_1)(q_1,q_2)$, we set $i_1=0$ and $i_2=2$.

Observe that once counter $j$ becomes positive, we can scale down the suffix of $\rho$ after $i_j$ by $\varepsilon=(|\rho||M|)^{-1}$, where $|M|$ is the maximal absolute counter update appearing in the VASS. Then, the effect of the suffix cannot make counter $j$ become negative, even if we remove cycles in the suffix. Moreover, removing cycles before $j$ becomes positive will keep the value of $j$ at 0, hence we can also remove cycles in the prefix of $\rho$ before $i_j$ without causing counter $j$ to become negative. It follows that for each counter (separately), after the appropriate scaling, we can remove any cycle from $\rho$ while keeping a valid run, with the exception of the cycle on which counter $j$ becomes positive, if it exists.

Therefore, by applying this reasoning to each of the $d$ counters, we scale $\rho$ so that any cycle can be removed except at most $d$ cycles corresponding to the $i_j$ indices. Thus, if $\rho$ has more than $|T|+d$ cycles, we can remove cycles while maintaining $\llbracket \pi \rrbracket = \llbracket \rho \rrbracket$, which would mean $\sigma$ is not of minimal length.
\end{proof}

%To prove the above, we observe that for any path $\pi$ there is at least one
%LPS $\sigma$ that satisfies the second and third constraints
%and is such that $\pi \in \sigma$. The latter implies $\Reach^{\vec{x}}(\pi) \subseteq \Reach^{\vec{x}}(\sigma)$ and thus the last
%constraint is also satisfied. Any decomposition of $\pi$ into simple paths and
% simple cycles will have the properties mentioned above. In our proof, we argue
% that any LPS $\sigma = \pi_0 \chi_1^* \dots \chi_n^* \pi_n$
% with minimal length $|\sigma|$ amongst those satisfying all but the length
% constraint, also satisfies $n \leq |T| + d$. Indeed, since we do not care
% about reaching $q(\vec{y})$ exactly, we can scale down the coefficients of any
% suffix of a run lifted from $\pi$. In turn, considering all indices where one
% of the counter values first becomes strictly positive --- so at most $d$
% indices --- this allows us to choose a run $\rho$ lifted from $\pi$ such that
% we can freely remove cycles from $\rho$ as long as the support constraint from
% the last item of the lemma is not violated.

Below we show that $\pi_2^*$ also admits a short LPS. 

\begin{lemma} \label{lem:LPS-cycle}
  Let $\chi$ be an arbitrary cycle such that $p(\vec{x}) \xrightarrow{\chi} q(\vec{y})$.
  Then, there exists a linear path scheme
  $\sigma =  \pi_0 \chi_1^* \pi_1 \dots \chi_n^* \pi_n$ such that:
  \begin{itemize}
    \item $n \leq 2(|T| + d + 1) + 1$,
    \item $\pi_i$ is a simple path for all $0 \leq i \leq n$,
    \item $|\chi_i| \leq 2|Q|(d|T|^{d+1}+1)$ for all $1 \leq i \leq n$, and
    \item $\vec{y} \in \Reach^{\vec{x}}_{\geq 0}(\sigma)$.
  \end{itemize}
\end{lemma}
\begin{proof}
We recall that the proof of the ``if'' direction of
  \autoref{lemma45} given in \cite{blondin2017logics} establishes
  that \autoref{itm:45itm1}--\autoref{itm:45itm4} together imply that $\vec{y} \in
  \Reach^{\vec{x}}_{\geq 0}(\pi_1 \pi_2^* \pi_3)$.
  By
  \autoref{lem:LPS-cover}, $\pi_1$ and $\pi_3$ admit short LPS $\sigma_1$ and $\sigma_3$. We show
  that $\pi_2^*$ also admits a short LPS of the form $\theta^*$, with $\theta$
  a cycle of length at most $2|Q|(d|T|^{d+1}+1)$. The result thus follows from
  \autoref{lemma45} and the length bounds from \autoref{lem:LPS-cover}.

  Let $\vec{z}$ be a vector reachable via the path $\pi_1$, i.e.,
  $q(\vec{x})\xrightarrow{\pi_1} q(\vec{z})$. By
  \autoref{prop:lp_scheme}, $\pi_2$ admits an LPS, so
  $\Reach^{\vec{z}}(\pi_2)\subseteq \Reach^{\vec{z}}(\sigma_2)$ for an LPS
  $\sigma_2=\tau_0 \gamma_1^* \tau_1\cdots \gamma_n^*\tau_n$, where $n \leq
  d|T|^{d+1}$, the $\tau_i$ are simple paths, and the $\gamma_i$ are simple
  cycles.  Assume w.l.o.g. that all the cycles in $\sigma_2$ are taken at
  least once in $\pi_2$ (otherwise we omit them from $\sigma_2$) and denote
  $\overline{\sigma_2}=\tau_0 \gamma_1 \tau_1\cdots \gamma_n\tau_n$. That is,
  $\overline{\sigma_2}$ is $\sigma_2$ where we take each cycle exactly once
  (note that this is a concrete path). We note that
  $\overline{\sigma_2}^*$ is an LPS comprising a single (not necessarily
  simple) cycle whose length is at most $2|Q|(d|T|^{d+1} +1)$ since it consists of
  at most $d|T|^{d+1}$ simple cycles and $d|T|^{d+1}+1$ simple paths. In order to conclude that $\vec{y} \in \Reach^{\vec{x}}_{\geq 0}(\sigma_1 \overline{\sigma_2}^* \sigma_3)$, it thus suffices to prove \autoref{lem:flatten-lps} below.
  \end{proof}
  
  \begin{restatable}{lemma}{lemflattenlps}\label{lem:flatten-lps}
    Let $\sigma$ be an LPS and $\vec{z} \in \mathbb{Q}^d$. Then,
  $\Reach^{\vec{z}}(\sigma^*)\subseteq
  \Reach^{\vec{z}}(\overline{\sigma}^*)$ and, moreover, $\Reach^{\vec{z}}_{\geq 0}(\sigma^*)\subseteq
  \Reach^{\vec{z}}_{\geq 0}(\overline{\sigma}^*)$.
  \end{restatable}
  
  \noindent
  Intuitively,
both hold since every iteration of $\sigma$
  repeats each cycle $\gamma_i$ a number of times $M_i$ with some
  coefficient $\alpha_i$. If $M_i\alpha_i\le 1$ for all $i$ in an
  iteration, this can be simulated in a single iteration of
  $\overline{\sigma}$ with the coefficient $M_i\alpha_i$ on $\gamma_i$.
  Otherwise, if $M_i\alpha_i> 1$ for some $i$, we can split that
  iteration into several ones of $\overline{\sigma}$ with all coefficients being halved, thus reducing
  the $M_i\alpha_i$ until all are at most $1$.

\begin{proof}[of \autoref{prop:LPS-nonneg}] 
We remark that an arbitrary path $\pi$ can always be decomposed into simple paths separated by at most $|Q|$ (not necessarily simple) cycles, each with different starting (and thus end) states. Now, for each cycle $\chi_i$, we apply \autoref{lem:LPS-cycle} to obtain a short LPS $\theta_i$ for it. Finally, by \autoref{lem:flatten-lps}, we can use $\overline{\theta_i}^*$ for it in our LPS. The bounds on the lenghts of the cycles then follow from the bounds in \autoref{lem:LPS-cycle}.
\end{proof}

In the next section, we turn to solving the reachability problem for LPS.

\subsection{Reachability for linear path schemes is tractable}
As in the $\mathbb{Q}$VASS case, here we are able to prove that determining whether a configuration is reachable via an LPS is decidable in polynomial time.
\begin{theorem}\label{thm:nonneg-lps-ptime}
    Given a linear path scheme $\sigma$ and $\vec{x},\vec{y} \in \mathbb{Q}_{\geq 0}^d$,
    determining whether $\vec{y} \in \Reach^{\vec{x}}(\sigma)$ is in \P.
\end{theorem}
Once more, our argument relies on an encoding into a system of CSH clauses. However, in contrast to $\mathbb{Q}$VASS, the encoding is slightly less elegant. For paths, instead of encoding an (affine) zonotope into a system of linear inequalities, we focus directly on the path $\pi = (p_1,p_2) \dots (p_m,p_{m+1})$ that gives rise to it.

Let $\vec{b_i} \coloneqq \ell(p_i,p_{i+1})$ for all $1 \leq i \leq m$. Now, consider the system $\vec{A'} \vec{z} = \vec{a} \land \vec{B}\vec{z} \leq \vec{b}
\land \vec{C'}\vec{z} < \vec{c}$, where $\vec{A}$, $\vec{a}$, $\vec{B}$, $\vec{b}$, $\vec{C}$, and $\vec{c}$ as defined in \autoref{eqn:a-paths}, \autoref{eqn:b-paths}, and \autoref{eqn:c-paths}, 
and $\vec{A'} = \begin{pmatrix}\vec{A} & \vec{I}\end{pmatrix} \in \mathbb{Q}^{d \times m+2d}$, $\vec{C'} = \begin{pmatrix}\vec{C} & \vec{0} & \dots & \vec{0}\end{pmatrix} \in \mathbb{Q}^{d \times m + 2d}$
as defined just before \autoref{lem:enc-cones}. We extend this system by adding constraints to the effect of keeping all partial sums (i.e. all intermediate configurations) nonnegative. Concretely, to $\vec{A'}$ we add $md$ rows and columns to obtain a matrix $\vec{A''}$ additionally encoding
\(
    \bigwedge_{k=1}^d\bigwedge_{n=1}^m \left(\sum_{j=1}^n \vec{g_j} z_j\right) = z_{m+2d+nd+k}.
\)
We also pad the other matrices with $0$-columns to the right and the vectors with $1$-rows below for the dimensions to match again without adding more constraints. 

From the definitions, we immediately get the following result.
\begin{lemma}\label{lem:enc-nonnegpath}
    The system $\vec{A''z} = \vec{a'} \land \vec{Bz} \leq \vec{b} \land \vec{Cz} < \vec{c}$ of CSH clauses has a rational solution $(\vec{\alpha},\vec{y},\vec{x},\vec{\iota}) \in \mathbb{Q}_{\geq 0}^{m+2d+md}$ if and only if $p_1(\vec{x}) \xrightarrow{\pi} p_{m+1}(\vec{y})$.
\end{lemma}

We are now ready to prove the theorem.
\begin{proof}[\autoref{thm:nonneg-lps-ptime}]
By \autoref{lem:enc-nonnegpath}, it suffices to argue that cycles can also be encoded into a system of CSH clauses. For this, we make use of the ``if'' direction of \autoref{lemma45}. We use subsystems to encode \autoref{itm:45itm2} and \autoref{itm:45itm3}, just like we did for paths. Note that taking $\pi_2 = \pi_3 = \chi$ is sufficient in this case, so it is indeed an instance of reachability via a path. To encode \autoref{itm:45itm3}, we make use of a subsystem, like \autoref{eq:CSH_cones_pos}. Unfortunately, we cannot merely conjoin these three systems of CSH clauses as the cycle can also be taken $0$ times in $\chi^*$. Instead, we modify the right conjunct from \autoref{eq:CSH_cones} so that all coefficients from both paths as well as the cycle are either all $0$ or all of them are strictly positive. The encoding of $\pi_2$ and $\pi_3$ is adapted in the same way --- i.e. $\vec{Cz} < \vec{c}$ is replaced by the same modified conjunct as discussed previously.
For each cycle, we thus obtain a system of CSH clauses
that has a nonnegative rational solution $(\dots, \vec{y},\vec{x}, \dots)$ if and only if $p(\vec{x}) \xrightarrow{\chi^*} q(\vec{y})$.
Finally, \autoref{thm:nonneg-lps-ptime} holds since conjoining the encodings of all paths and cycles of the LPS yields a master system of CSH clauses of polynomial size.
\end{proof}

%% file: paper_sections/QReachZero.tex
\section{The Complexity of Reachability with Zero Tests}\label{sec:zero-tests}

In this section, we argue that solving the reachability problem for continuous VASS with \emph{zero tests} is \NP-hard, already for LPS of dimension $2$. For convenience, we state the result for $\mathbb{Q}$VASS. However, we note that the same proof establishes the result for $\mathbb{Q}_{\geq 0}$VASS. Recall that from dimension $4$ onward, and for general continuous $\mathbb{Q}_{\geq 0}$VASS, the problem is known to be undecidable~\cite{blondin2017logics}.

We start by formally defining the model. A continuous VASS $\mathcal{V}$ of dimension $2$ with zero tests is a tuple $(Q,t,\ell,Z_1,Z_2)$, where $\mathcal{V}' = (Q,t,\ell)$ is a continuous VASS and $Z_i \subseteq Q$ for $i=1,2$. A run $\rho = q_1(\vec{x}_1) \dots q_n(\vec{x}_n)$ of such a VASS is a run of $\mathcal{V}'$ such that, additionally, for all $1 \leq i \leq n$ we have that if $q_i \in Z_j$, for some $j \in \{1,2\}$, then $(\vec{x_i})_j = 0$. That is, any run that reaches a state in $Z_j$ must be such that the the value of the $j$-th counter is $0$ then.

\begin{theorem}\label{thm:hard-zero}
    For every $d \in \mathbb{N}, d \geq 2$, given a $\mathbb{Q}$VASS (or $\mathbb{Q}_{\geq 0}$VASS) of dimension $d$, and two configurations $p(\vec{x})$ and $q(\vec{y})$, determining whether $\p(\vec{x}) \xrightarrow{*} q(\vec{y})$ is \NP-hard, even for linear path schemes of dimension $2$.
\end{theorem}

We reduce the \textsc{ExactCover} problem, shown to be \NP-hard by Karp \cite{karp1972reducibility}, to our reachability problem. The \textsc{ExactCover} problem asks, given a set $X$ and a collection $C$ of subsets of $X$, whether there exists a subset $C'\subseteq C$ that is a partition of $X$, i.e. such that: 
$\bigcup_{c'\in C'} c' = X$, and $c'_1 \cap c'_2 = \emptyset$ for all $c'_1, c'_2 \in C'$.

\begin{proof}
We first reduce \textsc{ExactCover} to an intermediate problem. Consider a bijection from $X$ to the first $|X|$ prime numbers. To simplify notation, we henceforth assume the elements of $X$ are themselves prime numbers. Then, each element in $C$ corresponds to a set of primes. Now, the \textsc{PrimeCover} problem asks, given $T\in \mathbb{N}$, $X$ and $C$ as before, whether there exists a subset $C' \subseteq C$ such that $\prod_{c'\in C'}\prod_{p\in c'} p = T$. If $T = \prod_{p\in X} p$, then, by the prime factorization theorem, \textsc{PrimeCover} has a positive answer if and only if \textsc{ExactCover} does as well.

Now, we show how to encode an arbitrary instance of the \textsc{PrimeCover} problem into the reachability problem for $2$-dimensional LPS with zero tests. Let $d\in \mathbb{N}$ and $e\in\mathbb{N}$, we first introduce the \emph{multiplier} gadget depicted in \autoref{fig:mult-gadget}.
    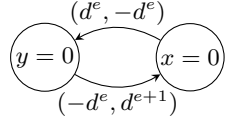
\begin{figure}[t]
        \centering
        \begin{tikzpicture}[inner sep=1pt, minimum size=0mm]
            \node (q0) [state] {$y = 0$};
            \node (q1) [state, right = of q0 ] {$x = 0$};
            \path [-stealth]
                (q0) edge [bend right] node [below] {$(-d^e, d^{e+1})$}   (q1)
                (q1) edge [bend right] node [above] {$(d^e, -d^e)$}   (q0);
        \end{tikzpicture}
        \caption{The multiplier gadget is shown; the state above is an element of $Z_1$ (noted by $x=0$) while the state below is an element of $Z_2$ (noted by $y=0$)}
        \label{fig:mult-gadget}
    \end{figure}
    We show that from all counter values $(x, 0) \in \mathbb{Q}_{\geq 0}^2$ it holds that if we follow the cycle in the gadget once then we reach $(dx, 0)$, as long as $x \leq d^{e-1}$. Indeed, to reach the right state, we can only choose the coefficient $\frac{x}{d^e}$, hence the counter values become $(0, \frac{xd^{e+1}}{d^e}) = (0, xd)$. Now, to get back to the left state, we can only choose the coefficient $\frac{x}{d^{e-1}}$, which results in the counter value $(xd, 0)$.
    
    Note that we have to make sure that both coefficients have to be at most $1$. This is ensured by our assumption that:
    \begin{equation}\label{eqn:ass-mult}
        x/d^{e-1} \leq 1 \iff x \leq d^{e-1}.
    \end{equation}
    Since the gadget is constructed in a way that for each transition there is only one feasible coefficient, if the assumption is not met, we cannot take the cycle. 
    
    Now, we can encode an instance of the \textsc{PrimeCover} problem as follows. For each set of primes $c \in C$ we create a multiplier gadget where $d = \prod_{p \in c} p$ and $e = \lceil \log_2(\prod_{c \in C}\prod_{p \in c} p) \rceil + 1$, and we link them in an LPS with transitions $(q_i,q_{i+1})$, for $1 \leq i \leq |C|$, labelled with $(0, 0)$ updates (see \autoref{fig:prime-cover}). We observe that all $d$'s and $e$ can be encoded in binary using a polynomial number of bits due to the prime number theorem. Further, we claim that \textsc{PrimeCover} has a positive answer if and only if $q_1(1,0) \xrightarrow{*} q_{|C|}(T,0)$ in the constructed LPS.
    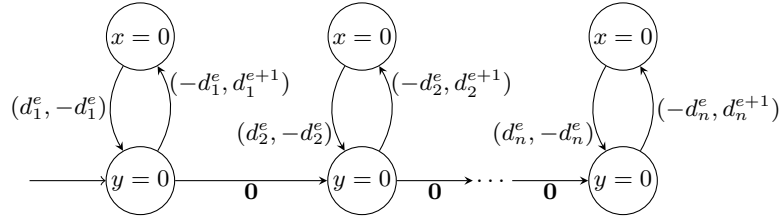
\begin{figure}[t]
        \centering
        
        \begin{tikzpicture}[yscale=1,inner sep=1pt, minimum size=0mm]
            \node (q0) [state, initial left,
                        initial text = {},
                        initial distance = 1cm,] {$y = 0$};
            \node (q1) [state, above = of q0 ] {$x = 0$};
            
            \node (q2) [state, right = 2cm of q0 ] {$y = 0$};
            
            \node (q3) [state, above = of q2] {$x = 0$};
            
            \node (ellipsis) [right = of q2] {$\dots$};
            
            \node (q4) [state, right = of ellipsis ] {$y = 0$};
            
            \node (q5) [state, above = of q4] {$x = 0$};
            
            \path [-stealth]
                (q0) edge [bend right] node [right,pos=0.8] {$(-d_1^e,d_1^{e+1})$}   (q1)
                (q1) edge [bend right] node [left] {$(d_1^e,-d_1^e)$}   (q0)
                (q0) edge node[below] {$\vec{0}$} (q2)
                (q2) edge [bend right] node [right,pos=0.8] {$(-d_2^e,d_2^{e+1})$}   (q3)
                (q3) edge [bend right] node [left,pos=0.8] {$(d_2^e,-d_2^e)$}   (q2)
                (q2) edge node[below] {$\vec{0}$} (ellipsis)
                (ellipsis) edge node[below] {$\vec{0}$} (q4)
                (q4) edge [bend right] node [right] {$(-d_n^e,d_n^{e+1})$}   (q5)
                (q5) edge [bend right] node [left,pos=0.8] {$(d_n^e,-d_n^e)$}   (q4);
        \end{tikzpicture}
        \caption{An LPS that encodes an instance of \textsc{PrimeCover} with $C = \{c_1, c_2, \dots, c_n\}$; recall that $d_i = \prod_{p \in c_i} p$}
        \label{fig:prime-cover}
    \end{figure}
    Indeed, if a subset $C' \subseteq C$ exists such that $\prod_{c' \in C'} \prod_{p \in c'} p = T$ then there is a run $\rho$ that takes the multiplier gadgets corresponding to all $C'$ exactly once. This is possible because any run that takes the multiplier gadgets at most once, up to the $n$-th gadget, does so with counter value $(x, 0)$ such that:
    \[
    \textstyle
    x \leq \prod_{c \in C}\prod_{p \in c} p \leq 2^{e-1} \leq p^{e-1} \text{ for all } p \in X
    \]
    because of our choice of $e$. Hence, we satisfy the assumption from \autoref{eqn:ass-mult} and $\rho$ can indeed take the aforementioned gadgets and reach $q_{|C|}$ with counter values $(T,0)$. Conversely, if such a run exists, by the prime factorisation theorem and the nature of the multiplier gadgets, we get that it takes a subset of the gadgets exactly once. The run thus induces a subset $C' \subseteq C$ as required.
\end{proof}

%% file: paper_sections/conclusion.tex
\section{Conclusion}
In this work, we give geometrical characterizations for the reachability sets of continuous VASS as well their flat and linear path scheme restrictions. Using these characterizations, we show that for any fixed dimension, polynomial-sized linear path schemes suffice as witnesses of reachability and that reachability in linear path schemes is tractable (even if the dimension is not fixed). In addition, sharpen hardness results for reachability for continuous VASS in the presence of zero tests: it is \NP-hard already for dimension two.

Two problems that remain open are the following. First, we do not know whether reachability is decidable for continuous VASS with zero tests of dimension $d =2,3$. In dimension $d=1$, the problem is known to be solvable in polynomial time~\cite{blmop21}. Second, we do not know what the exact complexity of the reachability problem is for fixed-dimension LPS with zero tests.

\subsection*{Acknowledgements}
We thank Michael Blondin, Filip Mazowiecki, and Igor Schittekat for useful conversations on the topics
of VASS and (continuous) variants thereof.

%% file: paper_sections/geometry_missing_proofs.tex
\section{Missing Proofs of Section~\ref{section:geometryRS}}
In this section, we provide the missing proofs of \autoref{section:geometryRS}.
Before the proofs, we introduce some notation not present in the paper. Let $G$ be the generator $G \coloneqq \{\ell(p_{i},p_{i+1}) \mid 1 \leq i \leq m \}$, then we also write $G \coloneqq \{\vec{g_1}, \vec{g_2}, \dots, \vec{g_n}\}$ to identify the generating vectors.

\thmcycles*
\begin{proof}
    We start by noting that $\{\vec{0}\}$ is included by definition in $ \relint(G) \cup \{0\}$ as well as $\Reach(\chi^*)$. This follows from the fact that the Kleene star allows us to not follow the cycle at all, thus reaching $\vec{0}$. Additionally, $\vec{0}$ is included in $ \relint(G) \cup \{0\}$ by definition.
    For the remainder of the proof, we focus on proving that $\Reach(\chi^*)\setminus \{\vec{0}\}$ is $\relint(G)$.  
    
    We first prove the $\subseteq$ direction. Here, we want to prove that for any $\vec{x} \in \Reach(\chi^*)\setminus\{\vec{0}\}$ it holds that $\vec{x} \in \relint(G)$. Since $\vec{0} \in \cone(G)$, we have $\aff(\cone(G)) = \linspan(G)$
    %\arkain{In case we want to give a short proof of this.
    %
    %Clearly $\aff(\cone(G)) \subseteq \aff(\linspan(G)) \subseteq \linspan(\linspan(G)) = \linspan(G)$.
    %Choose arbitrary $\sum_{i = 1}^n r_i \vec{g}_i \in \linspan(G)$.
    %$\vec{0} \in \cone(G)$ implies $\sum_{i = 1}^n r_i \vec{g}_i = \sum_{i = 1}^n r_i \vec{g}_i - (1 - \sum_{i = 1}^n r_i)\vec{0} \in \aff(\cone(G))$.
    %Hence $\linspan(G) \subseteq \aff(\cone(G))$.}
    and it thus suffices to show that there exists some $\varepsilon > 0$ s.t.
    $B_\varepsilon(\vec{x}) \cap \linspan(G) \subseteq \cone(B)$.
    Note that for any $\vec{x} \in \Reach(\chi^*)\setminus \{0\}$ there exists $\vec{\alpha} \in \mathbb{Q}^{|G|}_{>0}$ s.t. $\vec{x} = \sum_{i=1}^{|G|} \alpha_i \vec{g_i}$. It follows that for some $\delta > 0$ we have that both $\vec{x} + \delta \vec{g_i}$ and $\vec{x} - \delta \vec{g_i}$ are in $\cone(G)$. Let $H$ denote the convex hull of all such points, that is $H \coloneqq \convexhull(\{\vec{x} \pm \delta \vec{g_i} \mid \vec{g_i} \in G\})$. From our previous arguments, and by convexity of $H$ and, $\cone(G)$ we have that $H \subseteq \cone(G)$. It will thus suffice to prove that for some $\varepsilon > 0$ we have $B_\varepsilon(\vec{x}) \cap \linspan(G) \subseteq H$.
    
    Now, consider some \emph{basis} $B \subseteq G$ of $\linspan(G)$. That is, let $B$ be a maximal subset of $G$ such that all its elements are linearly independent. Further, consider a superset $D$ of $B$ such that $|D| = d$ and $D$ spans $\mathbb{Q}^d$. Then, we can set up the classical change-of-basis formula:
    \(
    \vec{z} = \vec{A} \vec{y},
    \)
    where $\vec{z},\vec{y} \in \mathbb{Q}^d$, $\vec{A} \in \mathbb{Q}^{d \times d}$ is a matrix whose columns correspond to the elements of $D$. Essentially, for any $\vec{z} \in \mathbb{Q}^d$, the unique values for $\vec{y}$ which make the formula true correspond to the (coefficients of the) representation of $\vec{z}$ with respect to the basis $D$ instead of the standard basis. In particular, for all $\vec{z} \in \linspan(G)$, since $B$ is a basis of it, $\vec{z}$ has a unique representation w.r.t. to it and thus all coordinates of $\vec{y}$ that correspond to elements of $D \setminus B$ are $0$. Now, since all columns of $A$ are linearly independent, the matrix has nonzero determinant and it is invertible. It follows that $\vec{A}^{-1}\vec{z} = \vec{y}$. Let $|M|$ denote the maximal absolute value of an entry of $\vec{A}^{-1}$ and choose $\varepsilon \coloneqq \sqrt{\delta/d|M|}$. It is easy to check that for any $\vec{z} \in \linspan(G)$, if $\lVert \vec{z} \rVert_2 < \varepsilon$ then every coordinate of $\vec{z}$ has absolute value strictly less than $d|M|\delta$. Now, from our choice of $|M|$, we have that every coordinate of $\vec{A}^{-1}\vec{z}$ is at most $\delta$. We thus conclude that $\vec{x} + \vec{z} \in H$, and since $\vec{z} \in \linspan(G), \lVert \vec{z} \rVert_2 < \varepsilon$ was arbitrary, we get that $B_\varepsilon(\vec{x}) \cap \linspan(G) \subseteq H$.
    %\arkain{In case we want to use a hammer.
    
    %The function $f : \mathbb{Q}^n \to \linspan(G)$ defined as $f(r_1,\dots,r_n) \coloneqq \sum_{i = 1}^n r_i \vec{g}_i$ is a bounded linear function and hence $\Reach(\chi^+) = f(\mathbb{Q}_+^n)$ is open in $\linspan(G)$ due to the open mapping theorem.
    %$\Reach(\chi^+) \subseteq \cone(G)$,
    %hence $\Reach(\chi^+) \subseteq \relint(G)$ since $\Reach(\chi^+)$ is open.}

    Now, we prove the $\supseteq$ direction. Here we want to prove that there exists $\vec{\alpha}\in \mathbb{Q}^{n}_{>0}$ s.t. $\vec{p} = \sum_{i=1}^n \alpha_i \vec{g_i}$, given that $\vec{p}$ is in $\relint(G)$. Since $\vec{p} \in \relint(G)$ it holds that there exists some $\varepsilon\in \mathbb{Q}$ s.t. for all $\vec{p'} \in \linspan(G)$ we have $\vec{p'}\in \cone(G)$ if $\lVert \vec{p} - \vec{p'} \rVert_2 < \varepsilon$. Hence, there exists some $\delta \in \mathbb{Q}_{>0}$ such that $\vec{p''} = \vec{p}-\delta \sum_{i=1}^n \vec{g_i} \in \cone(G)$. 
    %Since $\vec{p_1}\in \relint(G)$, we can repeat this argument to obtain $\vec{p_2} = \vec{p_1} - \delta_2 \vec{g_2}$. We can repeat this for every base in $G$ and obtain $\vec{p_n}$. 
    Moreover, since $\vec{p''} \in \cone(G)$, we have that there exists some $\vec{\beta}\in \mathbb{Q}^n_{\geq 0} $ such that:
    \[
        \vec{p''} = \sum_{i=1}^n \beta_i \vec{g_i}.
    \]
    By construction of $\vec{p''}$ we have that
    $\vec{p} = \vec{p''} + \delta\sum_{i=1}^n \vec{g_i}$.
    %\[
    %    \vec{p} = \vec{p_n} + \sum_{j=1}^n \delta_j \vec{g_j} = \sum_{i=1}^n \beta_i \vec{g_i}+ \sum_{j=1}^n \delta_j \vec{g_j}.
    %\]
    %
    Since $\beta_i \geq 0$ and $\delta_i > 0$ for all $1 \leq i \leq n$, we have that there exists $\vec{\alpha} \in \mathbb{Q}^n_{>0}$ (i.e. $\alpha_i = \beta_i + \delta$ for all $i$) s.t. $\vec{p} = \sum_{i=1}^n \alpha_i \vec{g_i}$. This concludes our proof.
\end{proof}
%\arkain{Why not use $\vec{p}' = \vec{p} -  \delta\sum_{i=1}^n\vec{g}_i$ and remove the induction.}

\simplebasis*
\begin{proof}
    Suppose there exists $(a_1,\dots,a_n) \in \mathbb{Q}_{>0}^n$ such that
    $\vec{x} = \sum_{i=1}^n a_i \vec{g_i}$. Let $A \coloneqq \max_{1\leq i
    \leq n} a_i$ and choose $D \in \mathbb{N}$ such that $1 \leq D$ and $A/D <
    1$. It remains for us to determine the value of $k$ and the corresponding
    coefficients. For the former, we set $k \coloneqq D$. For the latter, we
    will need some further notation: Let $M \coloneqq [\ell(p_i, p_{i+1}) \mid
    0 < i \leq n]$ be the multiset of all transition labels from $\chi$, and
    let $M(\vec{y})$ denote the multiplicity of $\vec{y}$ in $M$.
    For each $1 \leq i \leq nk$,
    we write $i'$ for the minimal natural number such that
    $i' \equiv i \pmod{n+1}$; and, $a(i)$ for the $a_j$ such that $\vec{g_j} =
    \ell(p_{i'},p_{i'+1})$.  Now, for each $1 \leq i \leq nk$, we set the
    coefficient $\alpha_i$ to the value:
    \[
        \frac{a(i')}{M(\ell(p_{i'},p_{i'+1}))D}.
    \]
    Hence, the constructed run ends with $p_1(\vec{x})$ as required.

    Conversely, suppose there exists a run from $p_1(\vec{0})$ to
    $p_1(\vec{x})$ with coefficients $\alpha_1, \alpha_2, \dots$ We define
    $(a_1,\dots,a_n)$ as required by summing the $\alpha_j$ as follows:
    \[
        a_i \coloneqq \sum_{j : \ell(q_j,q_{j+1}) = \vec{g_i}} \alpha_j,
    \]
    where $M$ once more denotes the multiset of all transition labels from
    $\chi$ and $M(\vec{x})$ the multiplicity of $\vec{x}$ in $M$. 
\end{proof}

\pathstozonotopes*
\begin{proof}

    By~\autoref{lem:simple-basis-path}, it follows trivially that no value outside $\zono(G_M)$ is reachable. It remains now to prove that all values in $\relint(G_M)$ are reachable, and all values in  $\boundary(\zono(G_M))$ are reachable iff they are in $\adj(G_M)$.

    We first show that for any point $\vec{p}\in \relint(G_M)$ there exists a vector $\vec{\alpha} \in \mathbb{Q}^n$ s.t. $0 < \alpha_i \leq 1$ for all $1 \leq i \leq n$ and $\vec{p} = \sum_{i=1}^n \alpha_i \vec{g_i}$. Since $\vec{p}\in \zono(G_M)$, we have that there exists a vector $\vec{\alpha} \in \mathbb{Q}^n$ such that $\vec{p}=\sum_{i=1}^n \alpha_i \vec{g_i}$ and $0\leq \alpha_i\leq 1$ for all $1 \leq i \leq n$. We call a vector $\vec{\alpha}$  a \emph{valid representation} of $\vec{p}$ if it satisfies the aforementioned constraints. Additionally, if $\vec{\alpha } > \vec{0}$ then the claim holds by \autoref{lem:simple-basis-path}. Henceforth, assume the complement holds and $\vec{p}$ has no strictly positive valid representation. Furthermore, we fix a valid representation $\vec{\alpha}$ such that the cardinality of the set $I \coloneqq \{ 1 \leq i \leq n \mid \alpha_i = 0\}$ is minimal among all representations of $\vec{p}$. It follows from our assumption that $I\neq \emptyset$. Since $\vec{p}\in \relint(G_M)$, we know there exists some $0 < \varepsilon < 1$ s.t. $\vec{p} + \varepsilon \vec{g}, \vec{p} - \varepsilon \vec{g} \in \zono(G_M)$ for all $\vec{g}\in G_M$. Hence, this property holds for all $\vec{g_k}\in G_M$ where $k\in I$. Let us focus on an arbitrary $k\in I$. First, we observe that $\vec{\alpha'} \coloneqq (\dots, \alpha_{k-1}, \varepsilon, \alpha_{k+1}, \dots)$ is a valid representation for $\vec{p} + \varepsilon \vec{g_k}$ and there exists some valid representation $\vec{\alpha''}$ for $\vec{p} - \varepsilon \vec{g_k}$. Then, we observe that $\vec{p}$ can be obtained as a convex combination of those two points:
    \[
        \vec{p} = \frac{1}{2}(\vec{p} + \varepsilon\vec{g_k}) +  \frac{1}{2}(\vec{p} - \varepsilon\vec{g_k}).
    \]
      
    Note that this new representation is valid since all coefficients are now at most $\frac{1}{2}$. Additionally, by construction of the representation of $\vec{p} + \varepsilon\vec{g_k}$ all coefficients that were non-zero are still non-zero. Finally, we have that the coefficient with index $k$ is also positive. It follows that the original representation for $\vec{p}$ was not minimal w.r.t. $I$. Hence, we arrive at our conclusion.

    Let $F$ be a face containing $\sigma_{G_M}$.
    Let $V$ be the set of vertices of $F$.
    Then, $\sigma_{G_M} \in F$ by definition of $\adj(G_M)$.
    Recall that $\aff(F)$ denotes the affine hull of $F$.
    Choose an arbitrary $\vec{x} \in \relint(F)$.
    Then, for $\lambda \in \mathbb{Q}$, define $\vec{y}_{\lambda} = \lambda \vec{x} + (1 - \lambda)\sigma_{G_M}$ and note that $\vec{y}_\lambda \in \aff(F)$.
    %$\vec{y}_{\lambda} = \vec{x} \in \relint(F)$.
    Furthermore, we have that $\lVert \vec{y}_{\lambda} - \vec{x}\rVert_2 \leq |1 - \lambda|(\lVert \vec{x}\rVert_2 + \lVert \sigma_{G_M}\rVert_2)$ by the triangle inequality of the norm $\lVert \cdot \rVert_2$.
    Since $\vec{x} \in \relint(F)$, using the definition of relative interior, we can say that there exists $\lambda > 1$ such that $\vec{y}_{\lambda} \in \relint(F)$.
    Note that $\vec{x} = \frac{1}{\lambda} \vec{y}_{\lambda} +  \left(1 - \frac{1}{\lambda}\right)\sigma_{G_M}$. Moreover,
    $\vec{y}_{\lambda} \in \relint(F) \subseteq F \subseteq \zono(G_M)$,
    so we can write $\vec{y}_{\lambda} = \sum_{\vec{g} \in G_M} r_{\vec{g}} \vec{g}$ for some $r_{\vec{g}} \in [0,1]$.
    Then
    \[
    \vec{x} = \sum_{\vec{g} \in G_M}
    \left(1 + \frac{r_{\vec{g}} - 1 }{\lambda}
    \right)\vec{g}
    \]
    Since $\lambda > 1 \geq r_{\vec{g}}$ for all $\vec{z} \in M$,
    we have $0 \leq \left(1 + \frac{r_{\vec{g}} - 1 }{\lambda}
    \right) \leq 1$.
    This implies $\vec{x} \in \Reach(\pi)$.
\end{proof}

\simplebasispath*
\begin{proof}
    Suppose there exists $(a_1,\dots,a_n) \in \mathbb{Q}^n$ as in the claim.
    We give coefficients to obtain a run lifted from $\pi = (p_1,p_2) \dots
    (p_m,p_{m+1})$. For this, recall that $M \coloneqq [\ell(p_i, p_{i+1})
    \mid 0 < i \leq n]$ is the multiset of all transition labels from $\pi$,
    and let $M(\vec{x})$ denote the multiplicity of $\vec{x}$ in $M$.
    Additionally, there is a mapping $\mu : \{\ell(p_i,p_{i+1}) \mid 1 \leq i
    \leq m\} \to G$ that essentially tells us how $G_M$ is obtained from $M$, i.e. $\mu(\vec{v})$ is the sum of all $\vec{v}$ from the multiset $M$ that have the same image per $\mu$.
    For each $1 \leq i \leq m$, we write $a(i)$ for the $a_j$ such that
    $\mu(\ell(p_i,p_{i+1})) = \vec{g_j}$; and, $N(i)$ for the cardinality of
    the set $\{\ell(p_i,p_{i+1}) \mid 1 \leq i \leq m,
    \mu(\ell(p_i,p_{i+1})) = \vec{g_j}\}$. Now, we set the coefficient
    $\alpha_i$ to the value:
    \[
      \frac{a(i)}{M(\ell(p_{i},p_{i+1}))N(i)}.
    \]
    By construction, $0 < a_i \leq 1$, for all $1 \leq i \leq n$, and
    the run ends with $p_m(\vec{x})$.

    We now prove the converse is also true. Consider any run $(p,\vec{0})
    \dots (q,\vec{y})$ lifted from $\pi$ with coefficients $a'_1,\dots,a'_m$
    such that $0 < a'_i \leq 1$ for all $1 \leq i \leq m$. For all $\vec{g_i}
    \in G$ and all $\vec{v} \in \mathbb{Q}^d$ such that $\mu(\vec{v}) =
    \vec{g_i}$, define $A^{(i)}_{\vec{v}}$ as the sum of all $a'_j$, for $1
    \leq j \leq m$, such that $\ell(p_j,p_{j+1}) = \vec{v}$.  (Note that the
    sum is nonempty because $\mu$ is surjective.) Recall that $\mu(\vec{v}) =
    \vec{g_i}$ implies that $\vec{v} \in \cone(\vec{g_i})$ and write
    $\beta^{(i)}_{\vec{v}}$ to denote the nonnegative value such that $\vec{v}
    = \beta^{(i)}_{\vec{v}} \vec{g_i}$. Now, let
    \(
        a_i \coloneqq \sum_{\vec{v} : \mu(\vec{v}) = \vec{g_i}}
        \beta^{(i)}_{\vec{v}}  A^{(i)}_{\vec{v}}.
    \)
    The following equalities hold.
    \begin{align*}
        \sum^n_{i=1} a_i \vec{g_i} ={} & \sum^n_{i=1} \left(\sum_{\vec{v} :
        \mu(\vec{v}) = \vec{g_i}} \beta^{(i)}_{\vec{v}}
        A^{(i)}_{\vec{v}}\right) \vec{g_i}\\
        {}={} & \sum^n_{i=1} \left(\sum_{\vec{v} : \mu(\vec{v}) = \vec{g_i}}
        (\beta^{(i)}_{\vec{v}}  \vec{g_i}) A^{(i)}_{\vec{v}}\right)\\
        {}={} & \sum^n_{i=1} \left(\sum_{\vec{v} : \mu(\vec{v}) = \vec{g_i}}
        \vec{ v} A^{(i)}_{\vec{v}}\right) & \text{def. of }
        \beta^{(i)}_{\vec{v}}\\
        {} = {} & \sum_{i=1}^m a'_i  \cdot \ell(p_i,p_{i+1}) = \vec{y} &
        \text{defs. of } \mu, A^{(i)}_{\vec{v}}
    \end{align*}
    It thus suffices to argue that $0 < a_i \leq 1$, for all $1 \leq i
    \leq n$. We observe that:
    \[
      \vec{g_i} = \sum_{\vec{v} : \mu(\vec{v}) = \vec{g_i}} M(\vec{v}) \vec{v}
      = \sum_{\vec{v} : \mu(\vec{v}) = \vec{g_i}} M(\vec{v})
      \left(\beta^{(i)}_{\vec{v}} \vec{g_i}\right)
    \]
    and thus $\sum_{\vec{v} : \mu(\vec{v}) = \vec{g_i}} M(\vec{v})
    \beta^{(i)}_{\vec{v}} = 1$. Finally, since we know $0 < a'_j \leq 1$, for
     all $1 \leq j \leq m$, then $0 < A^{(i)}_{\vec{v}} \leq M(\vec{v})$ and
    \(
        0 < \sum_{\vec{v} : \mu(\vec{v}) = \vec{g_i}} \beta^{(i)}_{\vec{v}}
        A^{(i)}_{\vec{v}} \leq 1.
    \)
\end{proof}

\geometrylps*
\begin{proof}
    The first summand follows directly from \autoref{thm:paths} and \autoref{lem:reachissum}. For the rest of the summands, we focus on an arbitrary $I_\ell$. Now, consider $\chi_i,\chi_j$ with $i,j \in I_\ell$. We have that the following equalities, where we write $\relint(G)$ instead of $\relint(\cone(G))$, hold.
    \begin{align*}
        & \Reach(\chi_i^*) + \Reach(\chi_j^*)\\
        {}={}& \relint(G(\chi_i)) \cup \{\vec{0}\} + \relint(G(\chi_j)) \cup \{\vec{0}\} & \text{by \autoref{thm:cycles}}\\
        {}={}&\left(\relint(G(\chi_i)) + \relint(G(\chi_j))\right) \\
        & {} \cup \relint(G(\chi_i)) \cup \relint(G(\chi_j)) \cup \{\vec{0}\} & \text{M. sum distributes over } \cup\\
        {}={}&\relint(G(\chi_i) \cup G(\chi_j)) \cup \{\vec{0}\} & \text{by \autoref{lemma:sumcones}} \\
    \end{align*}
    Note that the final equality uses the last part of \autoref{lemma:sumcones} which requires that $\linspan(G(\chi_i)) = \linspan(G(\chi_j))$. The summation grouping common linear subspaces thus follows by induction.
\end{proof}

\sumcones*
\begin{proof}
    We first note that $C+C' = \cone(G \cup G')$ follows immediately from the definition of convex cones. To prove that $\relint(C + C') = \relint(C) + \relint(C')$, we reuse the proof of \autoref{thm:cycles}. First, let us argue for the $\subseteq$-inclusion. Consider some $\vec{x} \in \relint(C + C')$. Then, from the first part of the proof of \autoref{thm:cycles}, we know that $\vec{x}$ has a representation with strictly positive coefficients w.r.t. to $G \cup G'$, i.e. $\vec{x} = \sum_{\vec{g} \in G \cup G'} \alpha_{\vec{g}} \vec{g}$ with $\alpha_{\vec{g}} > 0$ for all $\vec{g}$. Define: 
    \[
        \vec{c} = \sum_{\vec{g} \in G \cap G'} \frac{\alpha_{\vec{g}}}{2} \vec{g} + \sum_{\vec{g} \in G \setminus G'} \alpha_{\vec{g}} \vec{g}
    \]
    and similarly:
    \[
        \vec{c'} = \sum_{\vec{g} \in G' \cap G} \frac{\alpha_{\vec{g}}}{2} \vec{g} + \sum_{\vec{g} \in G' \setminus G} \alpha_{\vec{g}} \vec{g}.
    \]
    Note that $\vec{c} \in C$, $\vec{c'} \in C'$ and that, by construction, they both admit representations with strictly positive coefficients w.r.t. to $G$ and $G'$ respectively. It follows from the second part of the proof of \autoref{thm:cycles} that $\vec{c} \in \relint(C)$ and $\vec{c'} \in \relint(C')$. Since $\vec{x} = \vec{c} + \vec{c'}$ and $\vec{x}$ was arbitrary, the inclusion holds.
    
    Second, let us prove the $\supseteq$-inclusion is true as well. Consider $\vec{c} \in \relint(C)$ and $\vec{c'} \in \relint(C')$. Then, from the second part of the proof of \autoref{thm:cycles} we have that they admit a representation with positive coefficients w.r.t. to $G$ and $G'$ respectively. Clearly their sum $\vec{c} + \vec{c'} = \vec{x}$ admits such a representation w.r.t. $G \cup G'$. Hence $\vec{x} \in C + C'$. Moreover, because of the first part of the proof of \autoref{thm:cycles}, $\vec{x} \in \relint(C + C')$, thus concluding the argument.
    
    Now, suppose $\linspan(G) = \linspan(G')$. We show that $\relint(C) \subseteq \relint(C + C')$. The remaining inclusion is argued for identically.
    Let $V = \linspan(G) = \linspan(G') = \linspan(G \cup G')$ and note that, since $\vec{0} \in C \cap C' \cap
    \cone(G \cup G')$, we have:
    \[
    \aff(C) = \aff(C') = \aff(\cone(G \cup G')) = V.
    \]
    The following equalities thus hold.
    \begin{align*}
        & \relint(C)\\
        {}={} & \{\vec{c} \in C \mid \exists \varepsilon > 0, B_\varepsilon(\vec{c}) \cap V \subseteq C\}\\
        {}\subseteq{} & \{\vec{c} \in C \mid \exists \varepsilon > 0, B_\varepsilon(\vec{c}) \cap V \subseteq \cone(G\cup G')\} & C \subseteq \cone(G\cup G')\\
        {}\subseteq{} & \{\vec{c} \in \cone(G \cup G') \mid \exists \varepsilon > 0, B_\varepsilon(\vec{c}) \cap V \subseteq \cone(G\cup G')\} & \text{same as above}\\
        {}={} & \relint(\cone(G \cup G')) & \text{by def.}\\
        {} = {} & \relint(C + C')
    \end{align*}
    This concludes the proof.
\end{proof}

%% file: paper_sections/QReach_missing.tex
\section{Missing Proofs of Section~\ref{sec:qvass-reach}}

\cyclerep*

\begin{proof}
    By \autoref{thm:cycles}, it suffices to prove the following.
    \[
       \relint(\cone(G(\chi))) \cup \{\vec 0\} + \sum_{\theta \in C} \relint(\cone(G(\theta))) \cup \{\vec 0\} = \sum_{\theta \in C} \Reach(\theta^*)
    \]
    Now --- just like in the proof of \autoref{thm:geometry-lps} --- we have, by induction on \autoref{lemma:sumcones}, and because $\linspan(G(\theta_1)) = \linspan(G(\theta_2))$ for all $\theta_1,\theta_2 \in C \cup \{\chi\}$, that the following holds.
    \begin{equation}\label{eqn:ind-sumcones}
        \sum_{\theta \in C} \relint(\cone(G(\theta))) \cup \{\vec 0\} = 
        \relint\left(\cone\left(\bigcup_{\theta \in C} G(\theta)\right)\right) \cup \{\vec{0}\} 
    \end{equation}
    A second application of \autoref{lemma:sumcones} tells us the following.
    \begin{align*}
    & \relint(\cone(G(\chi))) \cup \vec{\vec{0}} +
    \relint\left(\cone\left(\bigcup_{\theta \in C} G(\theta)\right)\right) \cup \{\vec{0}\}\\
    {}  ={} & \relint\left(\cone\left(\bigcup_{\theta \in C \cup\{\chi\}} G(\theta)\right)\right) \cup \{\vec{0}\} 
    \end{align*}
    Observe that $\cone(\bigcup_{\theta \in C \cup \{\chi\}} G(\theta)) = \cone(\bigcup_{\theta \in C} G(\theta))$ because $\llbracket \chi \rrbracket \subseteq \bigcup_{\theta \in C} \llbracket \theta \rrbracket$ and therefore $G(\chi) \subseteq \bigcup_{\theta \in C} G(\theta)$. The result thus follows from \autoref{eqn:ind-sumcones} and \autoref{thm:cycles}.
\end{proof}

%% file: paper_sections/QposReach_missing_proofs.tex
\section{Missing Proofs of Section~\ref{sec:qposreach}} 

\lemflattenlps*
\begin{proof}
    In this proof, for a run $\rho$, we write $p(\vec{x}) \xrightarrow q(\vec{y})$ to denote the fact that the first and last configurations of $\rho$ are $p(\vec{x})$ and $q(\vec{y})$, respectively. For brevity, we give the proof only for the second part of the claim. The first part of the claim is much easier and follows from a subset of the arguments given below.

      We start with the following observation: let $\vec{u},\vec{v}$ be two vectors such that $q(\vec{u})\xrightarrow{\rho} q(\vec{v})$ for some run $\rho$, then: \[q(\vec{u})\xrightarrow{(\frac12\rho)^2} q(\vec{v})\] where $\frac12\rho$ is the run obtained from $\rho$ by multiplying all coefficients by $\frac12$. Indeed, assume by way of contradiction this does not hold, then: \[q(\vec{u})\xrightarrow{\frac12\rho} q(2^{-1}(\vec{v}+\vec{u}))\] but at least one of the counters, w.l.o.g. counter 1, becomes negative when following $\frac12\rho$ from $\frac12(\vec{v}+\vec{u})$. Let $b<0$ be the minimal value of counter 1 when following $\frac12\rho$ from counter $(0,0)$. The value $b$ is commonly referred to as the \emph{drop} or \emph{nadir} of $\frac12\rho$. Then,  $b<-\frac12(\vec{v}_1+\vec{u}_1)$, so $2b<-(\vec{v}_1+\vec{u}_1)$. Since $\vec{v}_1\ge 0$, we have in particular that $2b<-\vec{u}_1$. 
  Now notice that the drop of counter $1$ in $\rho$ is $2b$, meaning that $2b\ge -\vec{u}_1$, as $\rho$ can be taken with nonnegative counters from $\vec{u}$, which is a contradiction.
\end{proof}

%% file: paper_sections/encodingminskymachines.tex
\section{Zero-test reachability for LPS is in NP}
Here we provide the encoding of zero-test reachability for linear path schemes similar to the one given in \autoref{sec:qvass-reach}. We $\mathbb{Q}$VASS with zero tests of dimension $d \geq 1$ with zero-tests on the states. Our encoding yields an existential formula from the first-order logic over the integers with addition and order --- i.e., the structure $FO(\mathbb{Z}; +, 0, 1,\leq)$, a.k.a. Presburger arithmetic. It is well known that the satisfiability problem for the latter is in \NP{} (see, e.g.,~\cite{haase18}).

\begin{theorem}\label{theorem:reachminsky}
    Given a linear path scheme $\sigma$ and $\vec{x}, \vec{y} \in \mathbb{Q}^d$,
    determining whether $\vec{y} \in \Reach^{\vec{x}}(\sigma)$ is in \NP.
\end{theorem}

For the remainder of this section, we fix a $\mathbb{Q}$VASS with zero tests defined, using the notation from \autoref{sec:zero-tests}, as the tuple $(Q, T, \ell, Z_1, Z_2, \dots, Z_d)$.

% Explain well-foundedness and paths
Consider a path $\pi = (p_1, p_2) \dots (p_{n-1}, p_n)$, and $G(\pi) = \{\vec{g_1}, \vec{g_2}, \dots, \vec{g_n}\}$. We say that $\pi$ is \emph{can be taken from $\vec{x}$} if for all $1 \leq j \leq n$ we have that $p_j \in Z_k$ (for some $1 \leq k \leq d$) implies $(\vec{x})_k + \sum_{i = 1}^j (\vec{g_i})_k = 0$. It should be clear that $p(\vec{y})$ is
reachable from $q(\vec{x})$ via $\pi$ if and only if $\pi$ can be taken from $\vec{x}$ and $\vec{x} + \sum_{i = 0}^n \vec{g_i} = \vec{y}$. Further, note that being able to take a path from $\vec{x}$ can straightforwardly be encoded as a system of linear equalities over the integers in polynomial time.
\begin{lemma}\label{lemma:pathminsky}
    One can, in polynomial time, construct an existential Presburger sentence $\varphi$
    such that $\vec{z} \in \mathbb{Z}^{nd}$ is a solution of $\varphi$ if and only if $p_1(\vec{x}) \xrightarrow{\pi} p_n(\vec{y})$.
\end{lemma}

% Explain how cycles work
To conclude, we need to be able to encode zero-test rechability via cycles. Let us fix a cycle $\chi = (p_1, p_2) \dots (p_{n-1}, p_1)$. Observe that cycles can only be taken once, never, or arbitrarily many times from a given counter value $\vec{x}$. 

\begin{lemma}\label{lemma:numberofcyclesminksy}
    For all cycles $\chi$, exactly one of the following holds:
    $\chi$ cannot be followed from $\vec{x}$, $\chi$ can be followed once from $\vec{x}$, or $\chi$ can be followed arbitrarily many times from $\vec{x}$. 
\end{lemma}
\begin{proof}
    It is clear to see that if $\chi$ cannot be taken from $\vec{x}$, then we cannot take the cycle at all. Hence, we focus on the two other cases. Let $S \subseteq Q$, we write below $\chi \cap S$ to denote the set of states from $S$ that are reached along the cycle $\chi$.
    
    Let $\vec{x^{(1)}}$ be the counter values obtained after following the cycle once. Let $I \coloneqq \{i\in \mathbb{N} \mid \chi  \cap Z_i \neq \emptyset \}$, or the coordinates for which there is a zero test along $\chi$.
    We claim that if there exists some $i \in I$ s.t. $(\vec{x})_i \neq \vec{x^{(1)}}_i$, then we can follow the cycle at most once. Towards a contradiction, assume the complement holds and we can follow the cycle multiple times. Let $r$ be the first state along $\chi$ s.t. $r\in Z_i$. Since we can follow the cycle at least once, we know that $(\vec{x})_i + (\vec{g_1})_i + \dots + (\vec{g_r})_i = 0$, otherwise, the cycle cannot be taken once from $\vec{x}$. Similarly, we have that $(\vec{x}^{(1)})_i + (\vec{g_1})_i + \dots + (\vec{g_r})_i = 0$. But then $\vec{x}_i = \vec{x}^{(1)}_i$, which is in contradiction with our assumption. 
    
    Finally, we have to prove that in the remaining case we can take $\chi$ arbitrarily many times. This follows from the argument above and the fact that $\chi$ can be taken from $\vec{x}$. Since $(\vec{x})_i = \vec{x}^{(1)}_i$ for all $i\in I$, we have that the effect on all such coordinates is $0$ and the $\chi$ can be taken from $\vec{x}^{(1)}$. We can repeat this argument from $\vec{x}^{(1)}$ and thus, by induction, this concludes the proof.
\end{proof}

What remains is to argue that we can encode reachability via $\chi$ into an existential Presburger formula. By \autoref{lemma:numberofcyclesminksy}, it suffices to encode reachability via $\chi$ or $\chi^2$ as paths, using \autoref{lemma:pathminsky}. In the latter case, $\vec{y}$ can be $\vec{x}$ plus any factor of the effect of $\chi$. Finally, to know which case applies, we only need to check whether the effect of the cycle on certain coordinates (with zero tests) is $0$. This can clearly be encoded in the required fragment of Presburger arithmetic. 

\begin{lemma}
    One can, in polynomial time, construct an existential Presburger formula $\psi$ such that $\vec{z} \in \mathbb{Z}^{2nd}$ is a solution of $\psi$ if and only if $p_1(\vec{x}) \xrightarrow{\chi^*} p_n(\vec{x})$.
\end{lemma}

Combining the above with \autoref{lemma:pathminsky}, one obtains a (polynomial-sized) formula in the announced fragment of Presburger arithmetic for the whole LPS. This concludes the proof of \autoref{theorem:reachminsky}.